\documentclass[letterpaper,twoside,%
  twocolumn,aps,pra,floatfix,showkeys]{revtex4}

\usepackage{mathrsfs}
\usepackage{amsmath}
\usepackage{amsthm}
\usepackage{amssymb}
\usepackage{revsymb}
\usepackage{graphicx}
\usepackage{amsthm}
\usepackage{bbm}

\graphicspath{{../}}

\def\eqref#1{{(\ref{#1})}}

\def\diag{\mathop{\rm diag}}
\def\rank{\mathop{\rm rank}}
\def\mod{\mathop{\rm mod}}
\def\wgt{\mathop{\rm wgt}}

\def\Z{{\overline Z}}

\def\mat#1{\mathcal{#1}}
\def\Z2{\mathbb{Z}_2}

\newcommand{\ket}[1]{\left\vert{#1}\right\rangle}

\newtheorem{theorem}{Theorem}
\newtheorem{lemma}{Lemma}

\newtheorem{corollary}{Corollary}

\newtheorem{example}{Example}
\newtheorem{definition}{Definition}
\newtheorem{definitionX}{Definition}

\issuenumber{Issue Number}

\begin{document}

\title{Spin glass reflection of the decoding transition for quantum
  error correcting codes}

\author{Alexey A. Kovalev}

\affiliation{Department of Physics \& Astronomy  
and Nebraska Center for Materials and Nanoscience, 
University of Nebraska, Lincoln, Nebraska 68588, USA}

\author{Leonid P. Pryadko}

\affiliation{Department of Physics \& Astronomy, University of California,
  Riverside, California 92521, USA}

\begin{abstract}
  We study the decoding transition for quantum error correcting codes
  with the help of a mapping to random-bond Wegner spin models.
  Families of quantum low density parity-check (LDPC) codes with a
  finite decoding threshold lead to
  both known models (e.g., random bond Ising and random plaquette $\Z2$
  gauge models) as well as unexplored earlier generally non-local
  disordered spin models with non-trivial phase diagrams.  The
  decoding transition corresponds to a transition from the ordered
  phase by proliferation of extended defects which generalize the
  notion of domain walls to non-local spin models.  In recently
  discovered quantum LDPC code families with finite rates the number
  of distinct classes of such extended defects is exponentially large,
  corresponding to extensive ground state entropy of these codes.
  Here, the transition can be driven by the entropy of the extended
  defects, a mechanism distinct from that in the local spin models
  where the number of defect types (domain walls) is always finite.
\end{abstract}

\keywords{spin glass, Ising, Wegner, stabilizer code, decoding transition}
\maketitle

Locality in space-time is a great organizing principle for a
theoretical physicist who is trying to come up with a model for some
phenomenon.  It works beautifully both in high energy and in condensed
matter physics.  Depending on the details, the corresponding
techniques can be based on the derivative expansion, minimal gauge
coupling, or local lattice Hamiltonians.  Often enough, given a few
more specific symmetries and natural constraints, and considering only
the most local models, one can derive a unique functional form of an
effective Hamiltonian.

On the other hand, having concentrated for so long on local models, we
remain largely unaware of the physics that may be lurking out there,
beyond the familiar locality constraint.  Problem is, the space of
possible non-local continuum or discrete models is vast.  Without this
constraint, and given that most interactions in nature are indeed
local, what property can we use instead to select a non-trivial model?

In this work we explore disordered spin models associated with 
maximum likelihood decoding for stabilizer quantum error
correction codes\cite{gottesman-thesis,Calderbank-1997}.  The
construction is a generalization of the map between various surface
codes and two-dimensional spin
models\cite{Dennis-Kitaev-Landahl-Preskill-2002,%
  Landahl-2011,Katzgraber-Andrist-2013}.  This approach turns out both
physically intuitive and useful as a way to choose non-trivial spin
models, especially if one concentrates on quantum low density
parity-check (LDPC)
codes\cite{Postol-2001,MacKay-Mitchison-McFadden-2004}.

Unlike the case of the classical error-correcting codes\cite{MS-book}
where decoding can be done by minimizing certain energy
functional\cite{Sourlas-1989,Nishimori-book}, with a quantum
stabilizer code large groups of \emph{mutually-degenerate} errors can
not and need not be
distinguished\cite{gottesman-thesis,Calderbank-1997}.  To find the
most likely error, one has to decide between different equivalence
classes; this boils down to minimizing certain free energy functional
depending on the relevant error model.  We consider a particularly
simple error model where this functional can be readily interpreted as
the free energy for a disordered Ising spin model of a general form
studied by Wegner\cite{Wegner71} at some temperature $T\equiv
\beta^{-1}$ and individual bonds flipped independently with
probability $p$; the original decoding problem lives on the Nishimori
line\cite{Nishimori-1981,Nishimori-1980,Morita-Horiguchi-1980,%
  Nishimori-book} of the phase diagram, which generalizes the result
for the surface codes\cite{Dennis-Kitaev-Landahl-Preskill-2002}.  For
a code family where in the limit of large codes the decoding can be
done successfully with probability one, the corresponding spin models
are non-trivial, meaning that they definitely have an ordered
``defect-free'' phase at small $T$ and $p$, and a distinct disordered
phase at large $T$ and $p$.  In addition, in the clean limit, $p\to0$,
the spin models associated with quantum codes have exact self-duality,
in Wegner's sense\cite{Wegner71}.  Further, in the special case of
quantum LDPC codes\cite{Postol-2001,MacKay-Mitchison-McFadden-2004},
the Hamiltonians of the spin models is a sum of generalized Ising
bonds, each given by a (generally non-local) product of only a few
spin operators.

We show that the decoding transition corresponds to a transition from
the ordered phase by proliferation of extended defects which
generalize the notion of domain walls to non-local spin models.  In
the code families where the number of encoded qubits $k$ remains
finite in the limit of large $n$, the transition occurs when the
tension $\lambda$ of one or more of such defects vanish, akin to
vanishing line tension of a domain wall in the 2D Ising model.  In
quantum LDPC code families with finite
rates\cite{Tillich-Zemor-2009,Kovalev-Pryadko-2012,%
  Kovalev-Pryadko-Hyperbicycle-2013,%
  Andriyanova-Maurice-Tillich-2012,Bravyi-Hastings-2013} the number of
distinct classes of such extended defects is exponentially large,
corresponding to extensive ground state entropy of these codes.  Here,
the transition can happen even when all defects have finite energy
densities $\lambda\ge\lambda_0>0$, driven by the entropy of the
extended defects' types.  This mechanism is distinct from that in the
local spin models where the number of defect types (domain walls) is
always finite.

The paper is organized as follows.  We first give a brief overview of
the {\bf Main results}, namely, formulate all the theorems and
briefly describe other results, concentrating on the case of
Calderbank-Shor-Steane (CSS)
codes\cite{Calderbank-Shor-1996,Steane-1996}.  Then, we give a
detailed review of the necessary {\bf Background} facts about
quantum stabilizer codes.  In section {\bf Statistical
  mechanics of decoding} we describe how a spin model is constructed
from a given code family, and relate the possibility of successful
decoding with probability one to the existence of an ordered phase in
the corresponding spin model.  In section {\bf Phase
  transitions} 
we discuss the properties of the thermodynamical phase transition
corresponding to the maximum-likelihood (ML) decoding threshold,
introduce spin correlation functions which can characterize various
phases, and give several inequalities on the location of the
transition.
Finally, we give our {\bf Conclusions}.

\section{Main results}
We start by listing our main results.  To simplify the definitions,
here we concentrate on the case of CSS codes; more general results are
given later in the text along with the corresponding proofs.
\subsection{Definitions} 
\ A classical binary linear code\cite{MS-book}
$\mathcal{C}$ with parameters $[n,k,d]$ 
is a $k$-dimensional subspace of the vector space $\mathbb{F}_{2}^{n}$
of all binary strings of length $n$. Code distance $d$ is the minimal
weight (number of non-zero elements) of a non-zero string in the code.
Rows of the binary \emph{generator matrix} $G$ of the code
$\mathcal{C}\equiv \mathcal{C}_G$ are formed by its $k$ basis
vectors.  A linear code can 
also be specified by the binary \emph{parity check matrix}
$H$, $\mathcal{C}=\{\mathbf{c}\in\mathbb{F}_{2}^{n}|H\mathbf{c}^T=0\}$.
This implies that $H$ and $G$ are mutually orthogonal, $H
G^T=0$, and also 
\begin{equation}
  \label{eq:dual-code}
  \rank H+\rank G=n.
\end{equation}
Parity check matrix is a generating matrix of the code ${\cal
  C}^\perp
={\cal C}_H$ \emph{dual} to ${\cal C}$.  Respectively, the matrix $H$
is an exact dual to $G$, $H\equiv G^*$.  Note that here and throughout
this work we assume that all linear algebra is done modulo $2$, as
appropriate for the vector space $\mathbb{F}_2^n$.

Given a binary matrix $\Theta$ with dimensions $N_s\times
N_\mathrm{b}$, we define a generalized Wegner-type \cite{Wegner71}
partition/correlation function with multi-spin bonds $R_b\equiv
\prod_r S_r ^{\Theta_{r,b}}$ corresponding to the columns of $\Theta$
and Ising spin variables $S_r=\pm1$, $r=1,\ldots,N_s$:
\begin{equation}
  \label{eq:generalized-wegner}
  \mathscr{Z}_{\mathbf{e},\mathbf{m}}(\Theta;\{K\})\equiv {1\over 2^{N_g}}
  \sum_{\{S_r=\pm1\}}\prod_{b=1}^{N_\mathrm{b}}R_b^{m_b}{\exp \biglb(K_b
    (-1)^{e_b} R_b\bigrb) \over 2\cosh\beta},
\end{equation}
where we assume the couplings to be positive, $K_b\equiv \beta J_b>0$,
with $\beta$ being the inverse temperature, the
length-$N_\mathrm{b}$ binary vectors $\mathbf{e}$, $\mathbf{m}$
respectively specify the electric and magnetic disorder, and
$N_g\equiv N_s-\rank\Theta$ is the count of linearly-dependent rows in
$\Theta$.

A quantum CSS code\cite{Calderbank-Shor-1996,Steane-1996} with
parameters $[[n,k,d]]$ can be specified in terms of two $n$-column
binary generator matrices ${\cal G}_X$, ${\cal G}_Z$ with mutually
orthogonal rows, ${\cal G}_X {\cal G}_Z^{T}=0$.  Such a code encodes
$k=n-\rank {\cal G}_X-\rank {\cal G}_Z$ qubits in a block of $n$
qubits.  A CSS code can be thought of as a couple of binary codes, one
correcting $X$-type errors and the other $Z$-type errors.  However, it
turns out that any two errors $\mathbf{e}$ and $\mathbf{e}'$ of, e.g.,
$Z$-type differing by a linear combination of rows of ${\cal G}_Z$
have exactly the same effect on the quantum code---such errors are
called \emph{degenerate}.  The corresponding \emph{equivalence} is denoted
$\mathbf{e}\simeq\mathbf{e}'$.  A \emph{detectable} $Z$-type error
$\mathbf{e}=\mathbf{e}_Z$ has a non-zero \emph{syndrome}
$\mathbf{s}_Z={\cal G}_X \mathbf{e}^T$.  An undetectable and
\emph{non-trivial} $Z$-type error has a zero syndrome and is not degenerate
with an all-zero error; we will call such an error a (non-zero $Z$-type)
\emph{codeword} $\mathbf{c}=\mathbf{c}_Z$.  The distance $d$ of a CSS
code is the minimal weight of a $Z$- or an $X$-type codeword.

For each error type, we introduce a spin glass partition function:
\begin{equation}
  \label{eq:Z0}
  Z_0^{(\mu)}(\mathbf{e};\beta)\equiv\mathscr{Z}_{\mathbf{e},\mathbf{0}}({\cal
    G}_\mu;\{K_b=\beta\}),\;\mu=X,Z. 
\end{equation}
The normalization is such that for a model of independent $X$ or $Z$
errors with equal probability $p$ (probabilities of $e_b=1$ are
independent of each other and equal to $p$), at the Nishimori line
\cite{Nishimori-1981,Nishimori-1980,Morita-Horiguchi-1980,Nishimori-book}, 
\begin{equation}
  \label{eq:nishimori-temperature}
\beta=\beta_p,\quad  e^{-2\beta_p}=p/(1-p), 
\end{equation}
the partition function\ \eqref{eq:Z0} equals to a total
probability of a $\mu$-type error equivalent to $\mathbf{e}$.
We also define the partition function with an \emph{extended defect}
of additionally flipped bonds at the support of the codeword $\mathbf{c}$,
\begin{equation}
  \label{eq:Zc}
  Z_{\bf c}^{(\mu)}(\mathbf{e};\beta)\equiv
Z_{0}^{(\mu)}(\mathbf{e}+\mathbf{c};\beta),
\end{equation}
as well as the partition function corresponding to all  errors with
the same syndrome $\mathbf{s}$ as $\mathbf{e}\equiv \mathbf{e}_\mathbf{s}$,
\begin{equation}
  \label{eq:Ztot}
  Z_{\rm tot}^{(\mu)}(\mathbf{s};\beta)\equiv
\sum_\mathbf{c} Z_{\bf c}^{(\mu)}(\mathbf{e}_\mathbf{s};\beta)=\mathscr{Z}_{\mathbf{e}_\mathbf{s},\mathbf{0}}({\cal
    G}_{\bar\mu}^{*};\{K_b=\beta\}),
\end{equation}
where the summation is over all $2^k$ mutually non-degenerate
$\mu$-type codewords $\mathbf{c}$, such that ${\cal
  G}_{\bar\mu}\mathbf{c}^T=0$, and ${\bar\mu}=X$ if $\mu=Z$ and vice
versa.  The second form uses a matrix ${\cal G}_{\bar\mu}^*$ exactly
dual to ${\cal G}_{\bar\mu}$, cf.\ Eq.\ \eqref{eq:dual-code}.  Note that
Eq.\ \eqref{eq:Ztot} at $\beta=\beta_p$ gives the correctly normalized
probability to encounter the syndrome $\mathbf{s}$,
$\sum_\mathbf{s}Z_\mathrm{tot}(\mathbf{s};\beta)=1$.  Here
and below we omit the error-type index $\mu$ to simplify the
notations.  

Syndrome-based decoding is a classical algorithm to recover the error
equivalence class from the measured syndrome.  In maximum-likelihood
(ML) decoding, one picks the codeword
$\mathbf{c}=\mathbf{c}_\mathrm{max}(\mathbf{e})$ corresponding to the
largest contribution $Z_\mathrm{max}(\mathbf{s};\beta)\equiv
Z_{\mathbf{c}_\mathrm{max}(\mathbf{e})}(\mathbf{e};\beta)$ to the
partition function~\eqref{eq:Ztot} at $\beta=\beta_p$.  Given some
unknown error with the syndrome $\mathbf{s}$, the conditional
probabilities of successful and of failed ML recovery are, respectively,
\begin{equation}
  \label{eq:Psucc}
  P_\mathrm{succ}(\mathbf{s})
  ={Z_\mathrm{max}(\mathbf{s};\beta_p)\over
    Z_\mathrm{tot}(\mathbf{s};\beta_p)},
  \quad P_\mathrm{fail}(\mathbf{s})=1-P_\mathrm{succ}(\mathbf{s}).
\end{equation}
The corresponding average over errors can be written as a simple sum
over allowed syndrome vectors,
\begin{equation}
  \label{eq:Psucc}
  P_\mathrm{succ}
  \equiv \left[{Z_\mathrm{max}(\mathbf{s}_\mathbf{e};\beta_p)\over
    Z_\mathrm{tot}(\mathbf{s}_\mathbf{e};\beta_p)}\right]
=\sum_\mathbf{s}Z_\mathrm{max}(\mathbf{s};\beta_p),
\end{equation}
where the square brackets $[\,\cdot\,]$ denote an average over the errors
$\mathbf{e}$.  For a given infinite family of CSS codes,
asymptotically certain ML decoding implies $
P_\mathrm{succ}^{(X)}\to1$ and $
P_\mathrm{succ}^{(Z)}\to1$ in the limit of large $n$.

In terms of the spin glass model \eqref{eq:Z0}, this corresponds to a
phase where in thermodynamical limit each likely disorder
configuration $\mathbf{e}$ corresponds to a unique defect
configuration $\mathbf{c}=\mathbf{c}_\mathrm{max}(\mathbf{e})$:
\begin{definition}{ (CSS)}
  A \emph{fixed-defect phase} of the spin glass model \eqref{eq:Z0}
  corresponding to an infinite family of CSS codes
  has
  \begin{equation}
    [Z_\mathrm{max}^{(\mu)}(\mathbf{s}_\mathbf{e};\beta)/ 
      Z_\mathrm{tot}^{(\mu)}(\mathbf{s}_\mathbf{e};\beta)]\to1,\quad
    n\to\infty.
    \label{eq:fixed-defect-CSS}  
\end{equation}
\end{definition}\ignorespaces
It is also useful to define a special case of such a phase where any
likely disorder configuration does not introduce any defects:
\begin{definition}{ (CSS)}
  A \emph{defect-free phase} of the spin glass model \eqref{eq:Z0}
  corresponding to an infinite family of CSS codes has
  \begin{equation}
    [Z_0^{(\mu)}(\mathbf{e};\beta)/ 
      Z_\mathrm{tot}^{(\mu)}(\mathbf{s}_\mathbf{e};\beta)]\to1,\quad
    n\to\infty.
    \label{eq:defect-free-CSS}  
\end{equation}
\end{definition}\ignorespaces

\subsection{Results:\ ordered phases}
\ We first prove that the only allowed ordered
phase on the Nishimori line is the defect-free phase:
\begin{theorem}
  For an infinite family of quantum stabilizer codes successful
  decoding with probability one implies that on the Nishimori line the
  corresponding spin model is in the defect-free phase, i.e., in any
  likely configuration $\mathbf{e}$ of flipped bonds the largest
  $Z_\mathbf{c}(\mathbf{e};\beta_p)$ corresponds to
  $\mathbf{c}_\mathrm{max}(\mathbf{e})=\mathbf{0}$.
  \label{th:Nishimori-line-defect-free}
\end{theorem}
Definitions \ref{def:fixed} and \ref{def:defect-free} are formulated
in terms of the average ratios of partition functions.  As an
alternative, we introduce the free energy increment associated with
adding an extended defect $\mathbf{c}$ to a most likely configuration
at the given disorder $\mathbf{e}$ with the syndrome
$\mathbf{s}= G_{\bar\mu} \mathbf{e}^T$,
\begin{equation}
  \label{eq:defect-free-energy-max-CSS}
  \Delta F_\mathbf{c}^{\mathrm{max},\mu}(\mathbf{\mathbf{s}};\beta)\equiv
  \beta^{-1}\log{ Z_\mathrm{max}^{(\mu)}(\mathbf{s};\beta)\over
    Z_{\mathbf{c}_\mathrm{max}(\mathbf{e})+\mathbf{c}}^{(\mu)}(\mathbf{e};\beta)
  }, \;\,\mu=X,Z. 
\end{equation}
We prove 
\begin{theorem}\label{th:divergent-defect-energy-max}
  For an infinite family of disordered spin models~\eqref{eq:Z0}
 (or Eq.\ \eqref{eq:partition0}), 
  in a fixed-defect
  phase the averaged over the disorder free energy increment for an
  additional defect corresponding to a non-trivial codeword
  $\mathbf{c}\not\simeq\mathbf{0}$ diverges at large $n$, $[\Delta
    F_\mathbf{c}^{\mathrm{max}}(\mathbf{s}_\mathbf{e};\beta)]\to \infty$.
\end{theorem}
In the defect-free phase, the relevant analogous quantity is the free energy
increment with respect to a given error $\mathbf{e}$,
\begin{equation}
  \label{eq:defect-free-energy-zero-CSS}
  \Delta F_\mathbf{c}^{(0,\mu)}(\mathbf{\mathbf{e}};\beta)\equiv
  \beta^{-1}\log{ Z_0^{(\mu)}(\mathbf{e};\beta)\over
    Z_{\mathbf{c}}^{(\mu)}(\mathbf{e};\beta)}.
\end{equation}
The corresponding average over disorder diverges in the defect-free
phase where $\mathbf{c}_\mathrm{max}(\mathbf{e})=\mathbf{0}$ for every
likely error configuration $\mathbf{e}$.  Then, the Theorem
\ref{th:Nishimori-line-defect-free} leads to
\begin{corollary}
  \label{th:corollary}
  On the Nishimori line, the disorder-averaged free energy increment
  $[\Delta F_\mathbf{c}^{(0)}(\mathbf{e};\beta_p)]$ corresponding to
  any non-trivial codeword $\mathbf{c}\not\simeq\mathbf{0}$ diverges at
  large $n$ for $p<p_c$, where $p_c$ is the error probability
  corresponding to the ML decoding transition on the Nishimori line.
\end{corollary}

\noindent
We also introduce a \emph{tension} 
\begin{equation}
\lambda_\mathbf{c}\equiv {[\Delta F_\mathbf{c}^{\rm
      max}]\over
  d_\mathbf{c}},
\quad d_\mathbf{c}\equiv
\min_{\boldsymbol\sigma}\wgt(\mathbf{c}+{\boldsymbol\sigma}
{\cal G}),\label{eq:line-tension-c}
\end{equation}
an analog of the
domain wall line tension for the extended defects,
and prove 
\begin{theorem}
  \label{th:tension-average}
  For disordered spin models \eqref{eq:Z0} (or Eq.\
  \eqref{eq:partition0}) corresponding to an infinite family of
  quantum codes with asymptotic rate $R=k/n$, in a fixed-defect phase,
  the defect tension ${\overline\lambda}$ averaged over all
  non-trivial defect classes at large $n$ satisfy the inequality
  $\beta{\overline\lambda}\ge R\ln 2$.
\end{theorem}

\subsection{Results: order parameter}
\ The spin models corresponding to families of quantum codes include
the analogs of regular Ising model (e.g., regular Ising model on
square lattice for the toric codes) as well as various gauge models,
see Example \ref{ex:gauge}.  In general, there is no local order
parameter that can be used for an alternative definition of the
ordered phase.  In addition, while an analog of \emph{Wilson loop}
operator can be readily constructed for these models and has the usual
low- and high-temperature asymptotics, it remains an open question
whether it can be used to distinguish between specific disordered phases.

However, we constructed a set of non-local \emph{indicator} spin
correlation functions which must all be asymptotically equal to one in
the defect-free phase, while some of them change sign in the presence
of extended defects.  Using these, and the standard inequalities from
the gauge theory of spin glasses, we prove the following bound on the
location of the defect-free phase (this is an extension of Nishimori's
result\cite{Nishimori-1980,Nishimori-1981} on possibly reentrant phase
diagram for Ising models):
\begin{theorem}
  \label{th:boundary} Defect-free phase cannot exist at any $\beta$
for $p$ exceeding that at the decoding transition, $p>p_c$.
\end{theorem} 

\subsection{Results:\ phase transition}
\ For zero-$R$ codes, the only mechanism of a continuous transition is
for $\lambda_\mathbf{c}$ to vanish for some set of codewords
$\mathbf{c}$.  On the other hand, for finite-rate codes, Theorem
\ref{th:tension-average} implies that there is also a possibility
that at the transition point the tension remains finite,
$\lambda_\mathbf{c}\ge \lambda_\mathrm{min}>0$, for every codeword
$\mathbf{c}$.  This corresponds to a transition driven by the entropy
of extended defects.

While generically the transition in models with multi-spin couplings
is of the first order, it is continuous along the Nishimori line since
the corresponding internal energy is known exactly and is a continuous
function of $p$.  Moreover, the specific heat remains finite at the
transition point along the Nishimori line since the same inequality as
for regular spin glasses applies\cite{Nishimori-1980,%
  Morita-Horiguchi-1980,Nishimori-1981,Nishimori-book},  
\begin{equation}
  \label{eq:specific-heat-bound}
  [C(p;\beta_p]\le  N_\mathrm{b} {\beta_p^2\over \cosh^2\beta_p},
\end{equation}
where $N_\mathrm{b}=2n$ for the model \eqref{eq:partition0}, and
$N_\mathrm{b}=n$ for the models \eqref{eq:Z0} corresponding to a half
of a CSS code each.  Thus, as in the usual spin models, we expect that
the transition point $p=p_c$ at the Nishimori line is a multicritical
point where several phases come together.

Spin models corresponding to non-CSS zero-rate families of stabilizer
codes are exactly self-dual.  The same is true for CSS codes where the
two generator matrices ${\cal G}_X$, ${\cal G}_Z$ can be mapped to
each other, e.g., by column permutations, as is the case for the toric
codes and, more generally, for the hypergraph-product (HP)
codes\cite{Tillich-Zemor-2009}.  For many such models, the transition
point at the Nishimori line can be obtained to a high numerical
accuracy using the strong-disorder self-duality
conjecture\cite{Nishimori-1979,Nishimori-Nemoto-2003,Nishimori-2007,%
  Nishimori-Ohzeki-2006,Ohzeki-Nishimori-Berker-2008,%
  Ohzeki-2009,Bombin-PRX-2012,Ohzeki-Fujii-2012}
\begin{equation}
  \label{eq:self-duality-disordered}
H_2(p_c)=1/2,
\end{equation}
where $H_2(p)\equiv -p\log_2p-(1-p)\log_2(1-p)$ is the binary entropy
function.  While strictly speaking, there is no exact self-duality in
the presence of disorder\cite{Aharony-Stephen-1980}, we have confirmed
numerically that this expression is also valid, at least
approximately, for several models constructed here, e.g., models with
bond structure as in 
Example~\ref{ex:882-18-12}.

However, for code families with finite rate, the decoding transition
must be below the Shannon limit
\begin{equation}
R\le 1-H_2(p).\label{eq:shannon-threshold-CSS}
\end{equation}
Thus, Eq.\ \eqref{eq:self-duality-disordered} must be violated for
$R\ge1/2$.  On general grounds,  we actually  expect it to fail for
any code family 
with a finite rate, $R>0$.

\section{Background}
\label{sec:background}

\subsection{Stabilizer codes}
\label{sec:background-codes}
\ An $n$-qubit quantum code\cite{shor-error-correct,gottesman-thesis,%
  Calderbank-1997,Nielsen-book} is a subspace of the $n$-qubit Hilbert
space $\mathbb{H}_{2}^{\otimes n}$.  The idea is to choose a subspace
such that a likely error shifts any state from the code to a
linearly-independent subspace, to be detected with a suitable set of
measurements.  Any error, an operator acting on
$\mathbb{H}_{2}^{\otimes n}$, can be expanded as a linear combination
of the elements of the $n$-qubit Pauli group $\mathscr{P}_{n}$ formed
by tensor products of single-qubit Pauli operators $X$, $Y$, $Z$ and
the identity operator $I$: $\mathscr{P}_{n}=i^m \{I,X,Y,Z\}^{\otimes
  n}$, where $m=0,1,2,3$.  A \emph{weight} of a Pauli operator is the
number of non-trivial terms in the tensor product.

An $n$-qubit quantum \emph{stabilizer code} $\mathcal{Q}$ $[[n,k,d]]$
is a $2^k$-dimensional subspace of $\mathbb{H}_2^{\otimes n}$, a
common $+1$ eigenspace of all operators in the code's
\emph{stabilizer}, an Abelian group
$\mathscr{S}\subset\mathscr{P}_{n}$ such that
$-\openone\not\in\mathscr{S}$.  The stabilizer is typically specified
in terms of its generators, $\mathscr{S}=\left\langle
  S_{1},\ldots,S_{n-k}\right\rangle $.  Any operator proportional to
an element of the stabilizer $\mathscr{S}$ acts trivially on the code
and can be ignored.  A non-trivial error proportional to a Pauli
operator $E\not\in\mathscr{S}$ is detectable iff it anticommutes with
at least one stabilizer generator $S_i$; such an error takes a vector
from the code, $\ket\psi\in{\cal Q}$, to the state $E\ket\psi$ from an
orthogonal subspace $E{\cal Q}$ where the corresponding eigenvalue
$(-1)^{s_i}$ is negative.  Measuring all $n-k$ generators $S_i$
produces the binary syndrome vector $\mathbf{s}\equiv
\{s_1,\ldots,s_{n-k}\}$.  Two errors (Pauli operators) that differ by
an element of the stabilizer and a phase, $E_2=E_1 S e^{i\phi}$,
$S\in\mathscr{S}$, are called mutually degenerate; they have the same
syndrome and act identically on the code.

Operators commuting with the stabilizer act within the code; they have
zero syndrome.  A non-trivial undetectable error $E$ is proportional
to a Pauli operator which commutes with the stabilizer but is not a
part of the stabilizer.  These are the operators that damage quantum
information; minimal weight of such an operator is the
distance $d$ of the stabilizer code.  A quantum or classical code of
distance $d$ can detect any error of weight up to $d-1$, and correct
up to $\lfloor d/2\rfloor$.  

A Pauli operator $E\equiv i^{m'}X^{\mathbf{v}}Z^{\mathbf{u}}$, where
$\mathbf{v},\mathbf{u}\in\{0,1\}^{\otimes n}$ and
$X^{\mathbf{v}}=X_{1}^{v_{1}}X_{2}^{v_{2}}\ldots X_{n}^{v_{n}}$,
$Z^{\mathbf{u}}=Z_{1}^{u_{1}}Z_{2}^{u_{2}}\ldots Z_{n}^{u_{n}}$, can
be mapped, up to a phase, to a binary vector $\mathbf{e}\equiv
(\mathbf{v},\mathbf{u})$.  A product of two Pauli operators
corresponds to a sum ($\mod2$) of the corresponding vectors. Two Pauli
operators commute if and only if the \emph{trace inner product} of the
corresponding binary vectors is zero,
 $\mathbf{e}_1\star \mathbf{e}_2\equiv
 \mathbf{u}_{1}\cdot\mathbf{v}_{2}+\mathbf{v}_{1}\cdot\mathbf{u}_{2}=0
 \bmod 2$. With this map, generators of a stabilizer group are mapped
to rows of the binary generator matrix
\begin{equation}
G=(G_{X},G_{Z}),\label{eq:generator-matrix}
\end{equation}
with the
condition 
that the trace inner product of any two rows vanishes
\cite{Calderbank-1997}.  This commutativity condition can be also
written as $G_X G_Z^T+G_Z G_X^T=0$.

For a more narrow set of CSS codes stabilizer generators can be chosen
so that they contain products of only $X_i$ or $Z_i$ single-qubit
Pauli operators.  The corresponding generator matrix has the form
\begin{equation}
  G=\diag({\cal G}_X,{\cal G}_Z),
  \label{eq:CSS}
\end{equation}
where the commutativity condition simplifies to ${\cal G}_{X}{\cal G}_{Z}^{T}=0
\bmod2$.  The number of encoded qubits is $k=n-\rank G$; for CSS codes
this simplifies to $k=n-\rank
{\cal G}_X-\rank {\cal G}_z$.

Two errors are mutually degenerate iff the the corresponding binary
vectors differ by a linear combination of rows of $G$,
$\mathbf{e}'=\mathbf{e}+\boldsymbol{\alpha} G$.  It is convenient to
define the conjugate matrix $\widetilde G\equiv (G_Z,G_X)$ so that
$G\star G^T\equiv G \widetilde G^T=0$.  Then, the syndrome of an error
$\mathbf{e}\equiv (\mathbf{v},\mathbf{u})$ can be written as the
product with the conjugate matrix, $\mathbf{s}=\widetilde G
\mathbf{e}^T$.  A vector with zero syndrome is orthogonal to rows of
$\widetilde G$; we will call any such vector which is not a linear
combination of rows of $G$ a non-zero \emph{codeword}
$\mathbf{c}\not\simeq\mathbf{0}$.  Two codewords that differ by a
linear combination of rows of $G$ are equivalent,
$\mathbf{c}_1\simeq\mathbf{c}_2$; corresponding Pauli
operators are mutually degenerate.  Non-equivalent codewords represent
different cosets of the degeneracy group in the binary code with the
check matrix $\widetilde G$.  For an $[[n,k,d]]$ code, any non-zero
codeword has weight $\wgt(\mathbf{c})\ge d$, and there are exactly
$2k$ \emph{independent} codewords which can be chosen to correspond to
$2k$ operators $\bar X_i$, $\bar Z_i$, $i=1,\ldots,k$ (with the usual
commutation relations) acting on the logical qubits.

\subsection{LDPC codes}
\ A binary low density parity-check (LDPC) code is a linear code with
sparse parity check matrix\cite{Gallager-1962,%
  Richardson-Shokrollahi-Amin-Urbanke-2001,%
  Richardson-Urbanke-2001,Chung-Forney-Richardson-Urbanke-2001}.
These have fast and efficient (capacity-approaching) decoders. Over
the last ten years classical LDPC codes have become a significant
component of industrial standards for satellite communications, Wi-Fi,
and gigabit ethernet, to name a few. 
\emph{Quantum LDPC codes}\cite{Postol-2001,MacKay-Mitchison-McFadden-2004}
are just stabilizer codes\cite{gottesman-thesis,Calderbank-1997}, but
with stabilizer generators which involve only a few qubits each
compared to the number of qubits used in the code.  Such codes are
most often degenerate: some errors have trivial effect and do not
require any correction. Compared to general quantum codes, with a
quantum LDPC code, each quantum measurement involves fewer qubits,
measurements can be done in parallel, and also the classical
processing could potentially be enormously simplified.

One apparent disadvantage of quantum LDPC codes is that, until
recently\cite{Bravyi-Hastings-2013},  there has been no
known families of such codes that have finite relative distance
$\delta\equiv d/n$ for large $n$.  
This is in contrast to regular
quantum codes where the existence of ``good'' codes with finite
asymptotic rates $R\equiv k/n$ and finite $\delta$ has been
proved\cite{Calderbank-Shor-1996,Feng-Ma-2004}.  With such latter
codes, and within a model where errors happen independently on
different qubits with probability $p$, for $p<\delta/2$ all errors can
be corrected with probability one.  On the other hand, many quantum
LDPC code families have a power-law
scaling of the distance with $n$, $d\propto n^{\alpha}$, with
$\alpha\le 1/2$. Examples include code families in
Refs.~\cite{Tillich-Zemor-2009,Kovalev-Pryadko-2012,%
  Kovalev-Pryadko-Hyperbicycle-2013,Andriyanova-Maurice-Tillich-2012};
a single-qubit-encoding code family suggested in
Ref.~\cite{Freedman-Meyer-Luo-2002} has the distance scaling as
$d\propto (n\log n)^{1/2}$.  

An infinite quantum LDPC code family with sublinear power-law distance
scaling has a finite error correction threshold, including the
fault-tolerant case where the measured syndromes may have errors, as
long as each stabilizer generator involves a limited number of qubits,
and each qubit is involved in a limited number of stabilizer
generators\cite{Kovalev-Pryadko-FT-2013}.  This makes quantum LDPC
codes the only code family where finite rate is known to coexist with
finite fault-tolerant error-correction threshold, potentially leading
to substantial reduction of the overhead for scalable quantum
computation\cite{Gottesman-overhead-2013}.

Note that the quantum LDPC codes in
Ref.~\cite{Bravyi-Hastings-2013} have finite rate and finite relative
distance, at the price of stabilizer generator weight scaling like a
power-law, $w\propto n^{\gamma}$, $\gamma\le 1/2$; it is not known
whether a fault-tolerant error-correction protocol  exists
for such codes.



An example of a large code family containing quantum LDPC codes is the 
hypergraph-product (HP) codes \cite{Tillich-Zemor-2009} generalizing
the toric code. Such a code can be constructed from two binary matrices,
$\mat{H}_{1}$ (dimensions $r_{1}\times n_{1}$) and $\mat{H}_{2}$
(dimensions $r_{2}\times n_{2}$), as a CSS code with the generator
matrices \cite{Kovalev-Pryadko-2012}
\begin{equation}
  {\cal G}_{X}=(E_{2}\otimes\mathcal{H}_{1},\mathcal{H}_{2}\otimes
  E_{1}),\;\,
  {\cal G}_{Z}=(\mathcal{H}_{2}^{T}\otimes\widetilde{E}_{1},
  \widetilde{E}_{2}\otimes\mathcal{H}_{1}^{T}).
\label{eq:Till}  
\end{equation}
Here each matrix is composed of two blocks constructed as Kronecker
products (denoted with ``$\otimes$''), and $E_{1}$,
$\widetilde{E}_{1}$, $E_{2}$, $\widetilde{E}_{2}$ are unit matrices of
dimensions given by $r_{1}$, $n_{1}$, $r_{2}$ and $n_{2}$,
respectively. Let us denote the parameters of classical codes using
$\mat{H}_i$, $\mat{H}_i^T$ as parity check matrices,
$\mat{C}^\perp_{\mat{H}_{i}}=[n_{i},k_{i},d_{i}]$,
$\mat{C}^\perp_{\mat{H}_{i}^{T}}=
[{\widetilde{n}}_{i},\widetilde{k}_{i},\widetilde{d}_{i}]$, $i=1,2$,
with the convention\cite{Tillich-Zemor-2009} that the distance
$d=\infty$ if the corresponding $k=0$.
Then the parameters of the HP code are
$n=n_{2}r_{1}+n_{1}r_{2}$, 
$k=k_{1}\tilde{k}_{2}+\tilde{k}_{1}k_{2}$ while the distance $d$
satisfies\cite{Tillich-Zemor-2009} a lower bound 
$d\ge\min(d_{1},d_{2},\widetilde{d}_{1},\widetilde{d}_{2})$ and two
upper bounds: if 
$\widetilde{k}_{2}>0$, then $d\le d_{1}$; if 
$\widetilde{k}_{1}>0$, then $d\le d_{2}$.

Particularly simple is the case when both binary codes are
\emph{cyclic}, with the property that all cyclic shifts of a code
vector also belongs to the code\cite{MS-book}.  A parity check matrix
of such a code can be chosen \emph{circulant}, with the first row 
using the coefficients of 
the \emph{check}
polynomial $h(x)\equiv c_{0}+c_{1}x+\ldots+c_{n-1}x^{n-1}$ which is a
factor of $x^n-1$.  Then, we can choose both circulant matrices
$\mathcal{H}_{1}$ and $\mathcal{H}_{2}$ in Eq.\ \eqref{eq:Till} square
$n_i\times n_i$, which gives a CSS code with the parameters
$[[2n_1n_2,2k_1k_2,\min(d_1,d_2)]]$.  In particular, the toric
codes\cite{kitaev-anyons,Dennis-Kitaev-Landahl-Preskill-2002} are
obtained when the circulant matrices $\mathcal{H}_{1}$,
$\mathcal{H}_{2}$ are generated by the polynomial $h(x)=1+x$, with
$k_i=1$ and $d_i=n_i$, $i=1,2$.

\section{Statistical mechanics of decoding.}
\subsection{Maximum likelihood decoding}
\ Let us consider one of the simplest error models, where the bit flip
and phase flip errors happen independently and with equal probability
$p$.  The corresponding transformation of the single-qubit density
matrix can be written as
\begin{equation}
\rho\mapsto p_{I}\rho+p_{x}X\rho X+p_{y}Y\rho Y+p_{z}Z\rho Z\:,
\label{eq:decoh}
\end{equation}
where $p_{I}=(1-p)^2$, $p_x=p_z=p(1-p)$, $p_y=p^2$.  After relabeling
the axes ($y\leftrightarrow z$) this can be interpreted in terms of
the amplitude/phase damping model with some constraint on the
decoherence times $T_1$, $T_2$.  Our goal, however, is not to consider
the most general case, but to construct a simple statistical
mechanical model.

For the uncorrelated errors described by the completely-positive
trace-preserving map~\eqref{eq:decoh}, the probability of an error
described by the binary vector $\mathbf{e}=(\mathbf{v},\mathbf{u})$
(see the \textbf{Background} section) is
\begin{equation}
  \label{eq:error-probability}
  P(\mathbf{\mathbf{e}})=\prod_{i=1}^{N_{\rm b}} p^{e_i}(1-p)^{1-e_i}=
p^w(1-p)^{N_{\rm b}-w}, 
\end{equation}
where $N_{\rm b}=2n$ and $ w\equiv \wgt(\mathbf{e})=
\wgt(\mathbf{v})+\wgt(\mathbf{u})$ is the regular binary weight.  Now,
with a stabilizer code, all degenerate errors have the same effect and
cannot be distinguished.  Thus, one considers the net probability of
an error  having the same effect as $\mathbf{e}$,
\begin{equation}
  P_0(\mathbf{e})={1\over 2^{N_{g}}}\sum_{\boldsymbol{\sigma}} p^w
  (1-p)^{N_{\rm b}-w},\quad w\equiv \wgt(\mathbf{e}+\boldsymbol{\sigma} G),
\label{eq:prob0}
\end{equation}
where the generator matrix $G$ (see Eq.\ \eqref{eq:generator-matrix})
has dimensions $N_s\times N_{\rm b}$ and non-zero $N_g\equiv N_s-\rank G$
allows $G$ to have some linearly-dependent rows, cf.\ Eq.\
\eqref{eq:generalized-wegner}.  The errors in Eq.\ \eqref{eq:prob0} are
exactly degenerate with $\mathbf{e}$ but they are not all the errors
having the same syndrome as $\mathbf{e}$.  It is thus convenient to
introduce the probability of an error equivalent to $\mathbf{e}$
shifted by a codeword $\mathbf{c}$,
\begin{equation}
  P_\mathbf{c}(\mathbf{e})\equiv P_0(\mathbf{e}+\mathbf{c}),
\label{eq:prob-c}
\end{equation}
and the total probability of an error with the syndrome
$\mathbf{s}\equiv \widetilde G\mathbf{e}^T$, 
\begin{equation}
  \label{eq:prob-tot}
    P_{\rm tot}(\mathbf{s})=\sum_{\mathbf{c}}P_\mathbf{c}(\mathbf{e}),
\end{equation}
where $\mathbf{e}$ is any vector that gives the syndrome $\mathbf{s}$,
and the summation is done over all $2^{2k}$ inequivalent codewords,
length $N_{\rm b}$ zero-syndrome vectors, $\widetilde G \mathbf{c}^T=0$, that
are linearly independent from the rows of $G$, see
the \textbf{Background} section.  When combined with the summation
over the degeneracy vectors generated by the rows of $G$, see Eqs.\
\eqref{eq:prob0} and \eqref{eq:prob-c}, the
summation in Eq.\ \eqref{eq:prob-tot} can be rewritten as that over all
zero-syndrome vectors,
\begin{equation}
  \label{eq:prob-tot-mod}
  P_\mathrm{tot}(\mathbf{s})=\!\!\!\!\sum_{\mathbf{x}:\widetilde G \mathbf{x}^T=0}\!\!\!\!p^w
    (1-p)^{N_{\rm b}-w},\;\, w\equiv \wgt(\mathbf{e}+\mathbf{x}).
\end{equation}
The probability \eqref{eq:prob-tot-mod} is normalized properly, so that the
summation over all allowed syndrome vectors gives 1,
\begin{equation}
  \label{eq:prob-tot-summed}
  \sum_\mathbf{s}P_\mathrm{tot}(\mathbf{s})=1.
\end{equation}

When decoding is done, only the measured syndrome
$\mathbf{s}$ is known.  
For \emph{maximum likelihood} (ML) decoding, the inferred error vector
corresponds to the most likely configuration given the syndrome.  To
find it, we can start with some error configuration $\mathbf{e}\equiv
\mathbf{e}_\mathbf{s}$ corresponding to the syndrome $\mathbf{s}$, and
find a codeword $\mathbf{c}=\mathbf{c}_\mathrm{max}(\mathbf{e})$ such
that the corresponding equivalence class $\mathbf{e}+\mathbf{c}$ has
the largest probability,
\begin{equation}
P_{\mathbf{c}_\mathrm{max}(\mathbf{e})}(\mathbf{e})= P_\mathrm{max}(\mathbf{s})\equiv
\max_\mathbf{c} P_\mathbf{c}(\mathbf{e}).\label{eq:P-max}
\end{equation}
Unlike the codeword $\mathbf{c}_\mathrm{max}(\mathbf{e})$ which
depends on the choice of $\mathbf{e}$, the maximum
probability $P_\mathrm{max}(\mathbf{s})$ depends only on the syndrome
$\mathbf{s}\equiv \widetilde G\mathbf{e}^T$.  The conditional
probabilities of successful and of failed recovery given some unknown
error with the syndrome $\mathbf{s}$ become
\begin{equation}
  \label{eq:P-recovery-of-s}
  P_\mathrm{succ}(\mathbf{s}) ={P_\mathrm{max}(\mathbf{s})\over
    P_\mathrm{tot}(\mathbf{s})},\quad
  P_\mathrm{fail}(\mathbf{s})\equiv 1-P_\mathrm{succ}(\mathbf{s}). 
\end{equation}
The net probability of successful recovery averaged over all
errors can be written as 
\begin{equation}
  \label{eq:prob-recovery}
  P_\mathrm{succ}\equiv [  P_\mathrm{succ}(\mathbf{s}_\mathbf{e})]
=\sum_\mathbf{s}P_\mathrm{max}(\mathbf{s}).  
\end{equation}
Here and in the following $[f(\mathbf{e})]\equiv
\sum_\mathbf{e}P(\mathbf{e}) f(\mathbf{e})$ denotes the averaging over
the errors with the probability \eqref{eq:error-probability}.  The
result in the r.h.s.\ was obtained by partial summation over all
errors with the same syndrome, cf.\ the syndrome
probability \eqref{eq:prob-tot}.

Asymptotically successful recovery with probability one for an
infinite family of QECCs implies that in the limit of large $n$,
$P_\mathrm{succ}\to1$ while $P_\mathrm{fail}\to0$.  Alternatively, in
this limit Eqs.\ \eqref{eq:P-recovery-of-s} and \eqref{eq:prob-recovery} give
\begin{equation}\label{eq:max-phase}
  \left[\dfrac{P_{\rm max}(\mathbf{s}_{\bf e})}%
      {P_{\rm tot}(\mathbf{s}_{\bf e})}\right]\to1. 
\end{equation}
Comparing
Eqs.\ \eqref{eq:prob-tot-summed} and \eqref{eq:prob-recovery}, we see
that asymptotically, for each error that is likely to happen, the sum
\eqref{eq:prob-tot} is dominated by a single term with
$\mathbf{c}=\mathbf{c}_\mathrm{max}(\mathbf{e})$.
We can state this formally as 
\begin{lemma}
  \label{lemma:upside-down}
  For an infinite family of quantum codes, successful decoding with
  probability one implies that asymptotically
  at large $n$,  the ratio
  $$
  r(\mathbf{e})\equiv {  P_\mathrm{max}(\mathbf{s}_\mathbf{e})\over
    P_\mathrm{tot}(\mathbf{s}_\mathbf{e})}= 
  {P_\mathrm{max}({\bf e})\over \sum_{\bf c} P_{\bf c}({\bf s}_{\bf
      e})}\to1.  
  $$
  for any error configuration $\mathbf{e}$ likely to happen.
\end{lemma}
\begin{proof}
  Note that $r(\mathbf{e})< 1$. Indeed, the summation in the denominator is over all
  $\mathbf{c}$, one of them equals
  $\mathbf{c}_\mathrm{max}(\mathbf{e})$ while the remaining terms are positive.  Now, let us choose an arbitrarily small
  $\epsilon>0$ and separate the errors into ``good'' where
  $1-r(\mathbf{e})<\epsilon$ and ``bad'' where $1-r(\mathbf{e})\ge
  \epsilon$.  Use the following Bayesian expansion for the successful
  decoding probability:
  \begin{equation}
    \label{eq:Bayesian-expansion}
P_\mathrm{succ}=(1-P_\mathrm{bad}) \left[
      r(\mathbf{e})\right]_\mathrm{good}\!\!\!\!\!\!\!
  +P_\mathrm{bad}\left[
      r(\mathbf{e})\right]_\mathrm{bad}\!\!\!\!\!\!,
\end{equation}
where the averaging in each term is limited to a particular type of
errors as indicated.  The first term can be bounded from above by
$1-P_\mathrm{bad}$, while the second one
by $P_\mathrm{bad}(1-\epsilon)$, which gives  
\begin{equation}
    P_\mathrm{succ}\le
    1-\epsilon P_\mathrm{bad}.
    \label{eq:P-succ}
\end{equation}
Since
$P_\mathrm{fail}=1-P_\mathrm{succ}\to0$ at large $n$, the probability $P_\mathrm{bad}$
can be made arbitrarily small by choosing large enough
$n$.
\end{proof}

Generally,  given an infinite family of codes, asymptotically certain
recovery is possible with sufficiently small $p<p_c\le1/2$, as well in
the symmetric region $p>1-p_c$, while it may not be a sure thing in
the remaining interval $p_c\le p\le 1-p_c$.  This defines the ML
decoding transition. 

\subsection{Random bond spin model}
\label{sec:spin-model}
\ Given the well-established parallel between Wegner's models and binary
codes\cite{Sourlas-1989,Nishimori-book}, it is straightforward to come
up with a spin model matching the probabilities defined in the
previous section.  We use the binary error $\mathbf{e}$ to introduce
the bond disorder using $J_b=(-1)^{e_b}$, and consider Wegner's
partition function \eqref{eq:generalized-wegner} with $\Theta=G$,
\begin{equation}
  \label{eq:partition0}
  Z_0(\mathbf{e};\beta)\equiv \mathscr{Z}_{\mathbf{e},\mathbf{0}}(G,\{K_b=\beta\}).
\end{equation}
The normalization is 
such that the probability in Eq.\ \eqref{eq:prob0} is recovered on the
Nishimori line \eqref{eq:nishimori-temperature},
\begin{equation}
  \label{eq:nishimori-line}
  P_0(\mathbf{e})=Z_0(\mathbf{e};\beta_p),\quad e^{-2\beta_p}=p/(1-p).
\end{equation}
To shorten the notations, we will omit the inverse temperature $\beta$
whenever it is not likely to cause a confusion, $
Z_0(\mathbf{e})\equiv Z_0(\mathbf{e};\beta)$, and use
$P_0(\mathbf{e})$ at the Nishimori line, $\beta=\beta_p$.

We also define the partition function with an \emph{extended defect}
of flipped bonds at the support of the codeword $\mathbf{c}$,
$Z_\mathbf{c}(\mathbf{e};\beta)\equiv
Z_\mathbf{0}(\mathbf{e}+\mathbf{c};\beta)$ [cf.\
  Eq.\ \eqref{eq:prob-c}], the corresponding maximum
$Z_\mathrm{max}(\mathbf{s};\beta)\equiv
Z_{\mathbf{c}_\mathrm{max}}(\mathbf{e};\beta)$ [the maximum is reached
  at
  $\mathbf{c}_\mathrm{max}\equiv\mathbf{c}_\mathrm{max}(\mathbf{e};\beta)$
  which may differ from that in Eq.\ \eqref{eq:P-max} depending on the
  temperature], as well as an analog of $P_\mathrm{tot}(\mathbf{s})$
[Eq.\ \eqref{eq:prob-tot}],
\begin{equation}
  \label{eq:Z-tot}
  Z_\mathrm{tot}(\mathbf{s};\beta)=\mathscr{Z}_{\mathbf{e},\mathbf{0}}(\widetilde G^*,\{K_b=\beta\}), 
\end{equation}
where the binary
matrix $\widetilde G^*$ is exactly dual to $\widetilde G$, namely
$\widetilde G^*\widetilde G^{T}=0$ and $\rank \widetilde G+\rank
\widetilde G^*=N_{\rm b}$ (cf.\ Eq.\ \eqref{eq:dual-code}),  and we used the fact that 
$\widetilde G^*$ is a generating matrix for all vectors $\mathbf{x}$
in Eq.\ \eqref{eq:prob-tot-mod}.

Except for disorder, the partition function \eqref{eq:Z-tot} is
related to Eq.\ \eqref{eq:partition0} by Wegner's duality transformation
\cite{Wegner71},
\begin{equation}
 \dfrac{2^{(N_g-N_{s})/2}\mathscr{Z}_{\mathbf{e},\mathbf{0}}(\Theta,\{K\})}{ \prod_{b}\sqrt{(\tanh 
    K_{b})^{2}+1}}=
\dfrac{2^{(N_g^*-N_{s}^{*})/2}\mathscr{Z}_{\mathbf{0},\mathbf{e}}(\Theta^*,\{K^{*}\})}{
  \prod_{b}\sqrt{(\tanh  K_{b}^{*})^{2}+1}},
\label{eq:duality}
\end{equation}
where bonds are defined by the columns of a $N_s^*\times N_\mathrm{b}$ binary
matrix $\Theta^*$ exactly dual to $\Theta$, see
Eqs.\ \eqref{eq:dual-code} and \eqref{eq:generalized-wegner}.
The dual model
has the same number of bonds, $N_{\rm b}^{*}=N_{\rm b}$, $N_s^*$ spins, and its
ground state degeneracy parameter $N_{g}^{*}=N_{s}^{*}-\rank
\Theta^{*}$.  The coupling parameters of mutually dual bonds  are
related by $\tanh K_{b}=\exp(-2K_{b}^{*})$.  The conjugation in
Eq.\ \eqref{eq:Z-tot} just
rearranges the order of bonds and therefore leaves the partition/correlation
function invariant, except for corresponding permutation of
bond-specific variables: coupling parameters $K_b$ and electric
and magnetic charges, 
\begin{equation}
  \mathscr{Z}_{\mathbf{e},\mathbf{m}}(\widetilde G^*,\{K\})
  =\mathscr{Z}_{\widetilde{\mathbf{e}},\widetilde{\mathbf{m}}}
  (G^*,\{\widetilde
  K\}). \label{eq:conjugation}
\end{equation}

We note in passing that the  binary matrices $\Theta$ and
$\Theta^*$ defining the mutually dual partition functions in
Eq.\ \eqref{eq:duality} can be also thought of as the generating
matrices of the two dual binary codes [Eq.\ \eqref{eq:dual-code}], with
some additional linearly dependent rows.  In fact, Wegner's duality
has been long known in the coding theory as the MacWilliams identities
between weight generating polynomials of dual
codes\cite{MacWilliams-1963,MS-book}.

For a CSS code with the generator matrix in the form \eqref{eq:CSS} the
partition function \eqref{eq:partition0} splits into a product of
those for two non-interacting models corresponding to matrices ${\cal
  G}_{X}$ and ${\cal G}_{Z}$,
see Eq.\ \eqref{eq:Z0}.  In addition, two models defined by 
${\cal G}_{X}$ and ${\cal G}_{Z}$ are dual to each other modulo logical
operators.  We can find the ground state degeneracies $2^{N_g^\mu}$,
$\mu=X,Z$, of
the corresponding models from $N_g^\mu=N_s^\mu-\rank {\cal G}_\mu$, where
$N_s^\mu$, $\mu=X,Z$ defines the number of rows 
in the matrix ${\cal G}_\mu$. For hypergraph-product codes in
Eq.\ \eqref{eq:Till} 
the ground state degeneracy  is given by\cite{Kovalev-Pryadko-2012} 
$N_g^X=\tilde{k}_{1}\tilde{k}_{2}$ and $N_g^Z=k_{1} k_{2}$.

\begin{example}
  \label{ex:CSS}
  For a CSS code with the check matrix 
\eqref{eq:CSS}, the partition
  function \eqref{eq:partition0} is a product of those for two
  mutually decoupled spin models defined by matrices $\Theta={\cal
    G}_X$ and $\Theta={\cal G}_Z$, respectively, see 
Eq.\ \eqref{eq:Z0}.  Since ${\cal G}_X
  {\cal G}_Z^T=0$, in the absence of disorder these models are
  mutually dual, modulo logical operators.
\end{example}
\begin{example}
  \label{ex:TZ-vanilla}
  HP codes in Eq.\ \eqref{eq:Till} are CSS codes.  In
  the special case $\mathcal{H}_{1}=\mathcal{H}_{2}^T$, the matrices
  ${\cal G}_X$ and ${\cal G}_Z$ can be mapped
  to each other by permutations of rows and columns; the two spin
  models \eqref{eq:Z0} are identical.  In the absence of disorder both
  models are self-dual, modulo logical operators.
\end{example}
\begin{example}
  \label{ex:extended-cyclic}
  Suppose matrices $\mathcal{H}_{1}$ and $\mathcal{H}_{2}$ in
  Eq.\ \eqref{eq:Till} are square and circulant, corresponding to two
  cyclic codes with generally different check polynomials $h_1(x)$ and
  $h_2(x)$.  Then the matrices ${\cal G}_X$ and ${\cal G}_Z$ can be
  mapped to each other by permutations of rows and columns, and thus
  in the absence of disorder the corresponding spin models
  \eqref{eq:Z0} are self-dual modulo logical operators.  This map is
  generally different from 
  that in the previous example.  This case has a nice layout on square
  lattice with periodic boundary conditions, with the horizontal and
  vertical bonds $R_b$ in Eq.\ \eqref{eq:generalized-wegner} formed
  according to the pattern of coefficients in the polynomials $h_1(x)$
  and $h_2(x)$.  In particular, with $h_1(x)=h_2(x)=1+x$, the
  hypergraph-product code is a toric code, while
  Eq.\ \eqref{eq:Z0} gives two mutually decoupled Ising models.
\end{example}

\begin{example}
  \label{ex:DT}
  Debierre and Turban \cite{Debierre-Turban-1983} suggested a model
  that corresponds to a CSS code in the previous example
  with the  check polynomials $h_{1}(x)=1+x$ and
  $h_{2}(x)=1+x+\ldots+x^{l-1}$ for some positive integer $l$.
  The two binary codes have $k_1=1$ (codewords are all-one or all-zero
  vectors), and, with $n_2$ divisible by $l$, $k_{2}=l-1$ ($2^{l-1}$
  codewords given by the repetitions of all length-$l$ even-weight
  vectors).  With $l=3$, each of the two equivalent spin
  models~\eqref{eq:Z0} have four ground states in a pattern of stripes
  given by the repetitions of the vectors $[1,1,0]$, $[0,1,1]$,
  $[1,0,1]$ or $[0,0,0]$.  A boundary between two distinct ground
  states produce a pattern of ``unhappy'' bonds that corresponds to an
  extended defect $\mathbf{c}$ in Eq.\ \eqref{eq:prob-tot}, see
  Fig.~1, Right.
\end{example}

\begin{figure}[tbp]
\centering \includegraphics[width=1.\columnwidth]{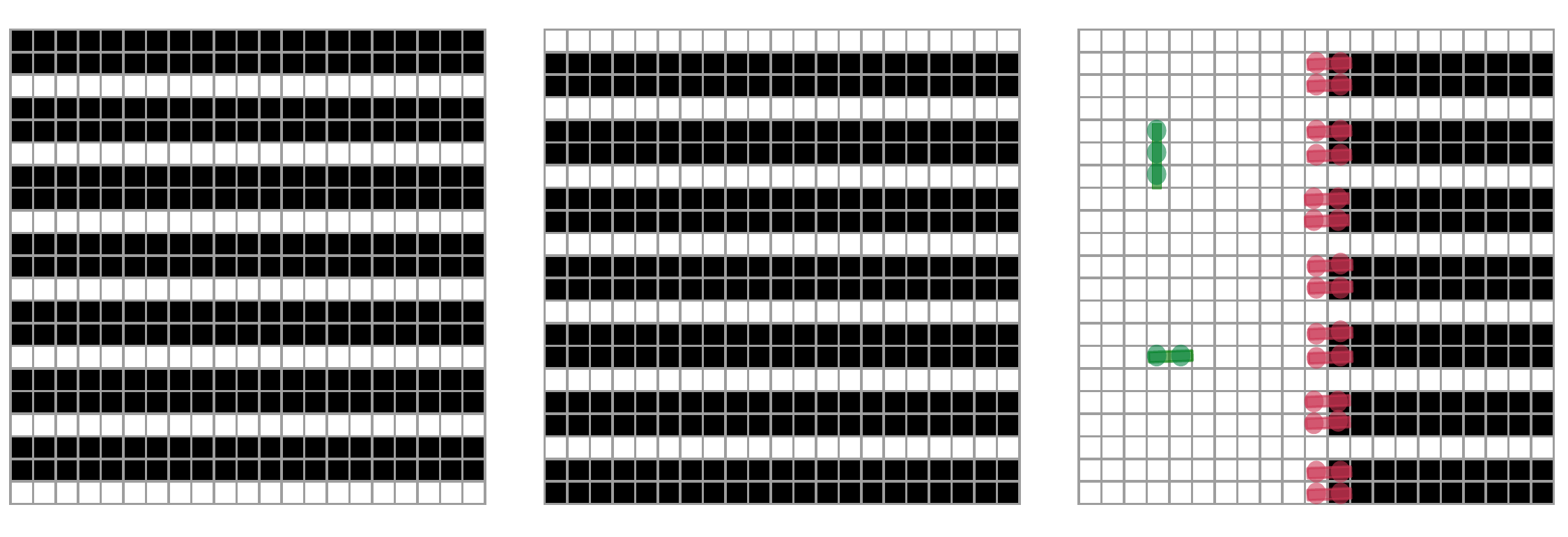} 
\caption{Left and Center: two basis ground states of the spin model in Example
  \protect\ref{ex:DT}, with black squares corresponding to flipped
  spins.  An arbitrary ground state of this spin model is a linear
  combination of these two.   Right: a domain wall between two
  such ground states.  Green squares show the pattern of vertical and
  horizontal bonds involving interactions of two or three spins,
  respectively.  A column of ``unhappy'' bonds forming the domain wall
  is shown with red.}
\label{fig:fig1} 
\end{figure}

\begin{figure}[tbp]
  \centering \includegraphics[width=1.\columnwidth]{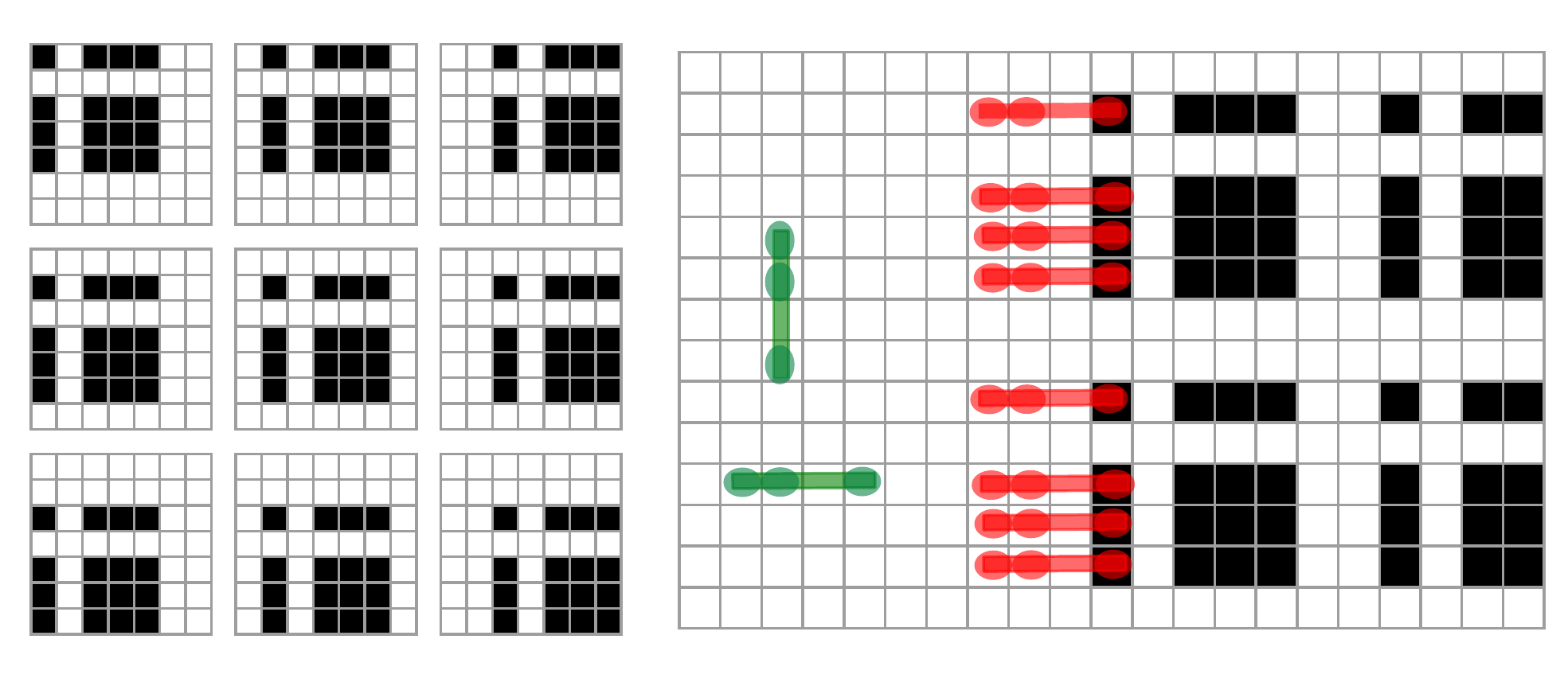} 
  \caption{Left: nine ground states of the spin model
    corresponding to the ${\cal G}_X$ matrix of the HP code~\eqref{eq:Till}
    generated by circulant matrices $\mathcal{H}_{i}$ corresponding to
    $h_{1}(x)=h_2(x)=1+x+x^3$, where $n_1=n_2=21$ (they both must be factors
    of $7$), see Example~\protect\ref{ex:882-18-12}. An arbitrary
    ground state of the spin model is a linear combination of these
    nine states.  Right: a domain wall formed between two such
    ground states.  Green squares on white background show the
    patterns of horizontal and vertical bonds, each involves three
    spins.  A column of ``unhappy'' bonds forming the extended defect
    is shown with red.}
  \label{figfigB}
\end{figure}

\begin{example}
  \label{ex:882-18-12}
  Spin models corresponding to quantum hypergraph-product codes
  $[[98s^2,6s,4s]]$, $s=1,2,\ldots$.  The model is constructed from
  $7s\times 7s$ circulant matrices $\mathcal{H}_{i}$ corresponding to
  $h_{i}(x)=1+x+x^3$, $i=1,2$.  A ground state of such a model is a
  linear combination of the nine basis states with the unit cell in
  Fig.\ 2, Left. 
  Fig.\ 2, Right: a boundary between two ground states.
\end{example}

\subsection{Ordered state}
\ In contrast to spin glass theory of classical binary codes where it is
generally possible to apply a gauge transformation so that perfect
decoding corresponds to a uniform
magnetization \cite{Sourlas-1989,Nishimori-book}, this is not
necessarily possible in the setting corresponding to a quantum code.
For example, the models with the partition function
\eqref{eq:partition0} include those with exact $S\to -S$ symmetry.  In
such a case it appears natural to introduce the average spin as an
order parameter.  On the other hand, such a symmetry is not generic;
the partition function \eqref{eq:generalized-wegner} may not even have
any degeneracy if $\Theta$ is a full-row-rank matrix.  Also, except
for the toric and related
codes local in 2D \cite{Dennis-Kitaev-Landahl-Preskill-2002}, it is not at all
clear what would be the relation of such an order parameter to the
decoding transition in a given code.

Here, we define an \emph{ordered} phase as an analog of the region of
parameters where asymptotically certain decoding is possible.  We
start with two definitions describing different phases: 
\begin{definitionX}
  \label{def:fixed}
  A \emph{fixed-defect phase} of the spin glass model \eqref{eq:partition0}
  corresponding to an infinite family of stabilizer codes has
  \begin{equation}
    [Z_\mathrm{max}(\mathbf{s}_\mathbf{e};\beta)/ 
      Z_\mathrm{tot}(\mathbf{s}_\mathbf{e};\beta)]\to1,\quad
    n\to\infty.
    \label{eq:fixed-defect}  
\end{equation}
\end{definitionX}
\begin{definitionX}
  \label{def:defect-free}
  A \emph{defect-free phase} of the spin glass model \eqref{eq:partition0}
  corresponding to an infinite family of stabilizer codes has
  \begin{equation}
    [Z_0(\mathbf{e};\beta)/ 
      Z_\mathrm{tot}(\mathbf{s}_\mathbf{e};\beta)]\to1,\quad
    n\to\infty.
    \label{eq:defect-free}  
\end{equation}
\end{definitionX}
We note that analogs of Lemma \ref{lemma:upside-down} apply for the
ratios in Eqs.\ \eqref{eq:fixed-defect} and \eqref{eq:defect-free}.
Thus, both in the fixed-defect and the defect-free phases, for any
error $\mathbf{e}$ likely to happen, the partition function
$Z_\mathrm{tot}(\mathbf{s}_\mathbf{e};\beta)$ is going to be dominated by a
single defect configuration, $\mathbf{c}_\mathrm{max}(\mathbf{e})$.
In the defect-free phase,
$\mathbf{c}_\mathrm{max}(\mathbf{e})=\mathbf{0}$, while in a
fixed-defect phase one may have a non-trivial defect
$\mathbf{c}_\mathrm{max}(\mathbf{e})\not\simeq\mathbf{0}$.  


\subsection{No fixed-defect phase on the Nishimori line} \ On the
Nishimori line, the definition of a fixed-defect phase matches that of
a region with asymptotically certain successful decoding, see Eq.\
\eqref{eq:max-phase}.  The latter region terminates at the decoding
transition at the single-bit error probability $p=p_c$.  On the other
hand, the proof of the lower bound on the decoding threshold from
Ref.~\cite{Kovalev-Pryadko-FT-2013} actually establishes the existence
of a zero-defect phase on the Nishimori line, for small enough $p$.
With both phases present, one would expect an additional transition
between these phases at some $p<p_c$.  Theorem
\ref{th:Nishimori-line-defect-free} on p.\
\pageref{th:Nishimori-line-defect-free} shows that this does not
happen because there is no fixed-defect phase along the Nishimori line.

\begin{proof} of {Theorem \ref{th:Nishimori-line-defect-free}.}
  Below the decoding transition, $p<p_c$, according to Lemma
  \ref{lemma:upside-down}, the probability
  $P_\mathrm{tot}(\mathbf{s})$ to obtain each likely syndrome is
  dominated by a single disorder configuration
  $\mathbf{e}_0(\mathbf{s})$.  This is also the configuration most
  likely to happen, as opposed to any other configuration
  corresponding to the same syndrome.  
\end{proof}
In comparison, for $\beta\neq\beta_p$, the disorder probability
distribution $P_0(\mathbf{e})$ is different from the partition
function $Z_0(\mathbf{e};\beta)$.  In general, the dominant
contribution to $Z_\mathrm{tot}(\mathbf{s}_\mathbf{e};\beta)$ may come
from some other defect configuration
$\mathbf{c}_\mathrm{max}(\mathbf{e};\beta)\not\simeq\mathbf{0}$.  

In practical terms, when designing a decoding algorithm, we can
concentrate on the portion of the free energy corresponding to
$Z_\mathbf{0}(\mathbf{e};\beta_p)$ and ignore the possibility of any
non-trivial defects without affecting the decoding probability in the
limit of large $n$.



\subsection{Free  energy of a defect}
\ {\em In a fixed-defect phase:\/} Let us introduce the free energy cost of
flipping the bonds corresponding to an non-zero bits of the codeword
$\mathbf{c}$ on top of the flipped bond pattern in the most likely
configuration $\mathbf{c}_\mathrm{max}(\mathbf{e})$ corresponding to
an error $\mathbf{e}$ with the syndrome $\mathbf{s}=\widetilde G
\mathbf{e}^T$,
\begin{equation}
  \label{eq:defect-free-energy-max}
  \Delta F_\mathbf{c}^{\mathrm{max}}(\mathbf{\mathbf{s}};\beta)\equiv
  \beta^{-1}\log{ Z_\mathrm{max}(\mathbf{s})\over
    Z_{\mathbf{c}_\mathrm{max}(\mathbf{e})+\mathbf{c}}(\mathbf{e}) }.
\end{equation}
\begin{proof}{of Theorem \ref{th:divergent-defect-energy-max}.} In the
  fixed-defect phase each syndrome $\mathbf{s}$ likely to 
  happen must be characterized by a unique configuration of defects,
  with the other configurations strongly suppressed.  Version of Lemma
  \ref{lemma:upside-down} appropriate for this phase (see Def.\
  \ref{def:fixed}) implies that $\Delta
  F_\mathbf{c}^\mathrm{max}(\mathbf{s};\beta)\to \infty$
  asymptotically at large $n$.  The corresponding disorder average must also
  diverge at large $n$.
\end{proof}
If we introduce the minimum weight $d_\mathbf{c}$ of a bit string in
the degeneracy class of $\mathbf{c}$, $d_\mathbf{c}\equiv
\min_{\boldsymbol{\sigma}}\wgt(\mathbf{c}+\boldsymbol{\sigma} G)$, we
can formulate the following bounds
\begin{lemma}
  \label{lemma:Fmax-bounds}
  For any error $\mathbf{e}$ which gives the syndrome
  $\mathbf{s}$, any codeword $\mathbf{c}$, and any temperature
  $\beta^{-1}$, $0\le \Delta
  F_\mathbf{c}^{\mathrm{max}}(\mathbf{s};\beta)\le 2d_\mathbf{c}$.
\end{lemma}
\begin{proof}
  The lower bound follows trivially from the fact that $Z_{\rm
    max}(\mathbf{s})$ is the largest of $Z_\mathbf{c}(\mathbf{e})$.
  To prove the upper bound, use the Gibbs-Bogoliubov inequality in the
  form:
  \begin{equation}
    \beta^{-1}\log{Z_0(\mathbf{e}')\over
      Z_\mathbf{c}(\mathbf{e}')}
    \le  \langle
    E_{\mathbf{c}+\mathbf{e}'}-E_{(\mathbf{e}')}\rangle=
    \sum_{b:\mathbf{c}_b\not\simeq\mathbf{0}}  
    2\langle (-1)^{\mathbf{e}'_b} R_b\rangle ,\label{eq:upper-bound}  
  \end{equation}
  where $\mathbf{e}' \equiv \mathbf{e}+{\mathbf{c}_{\rm
      max}(\mathbf{e})}$ is the same-syndrome disorder configuration
  such that the maximum is reached at $\mathbf{c}=\mathbf{0}$,
  $E_{\mathbf{e}}\equiv \sum_b (-1)^{e_b}R_b$ is the energy of a spin
  configuration, see Eq.\ \eqref{eq:generalized-wegner}, and the
  averaging is done over all spin configurations contributing to
  $Z_0(\mathbf{e}';\beta)$.  Each term in the r.h.s.\ of 
  Eq.\ \eqref{eq:upper-bound} is uniformly bounded from above, $2
  (-1)^{\mathbf{e}_b} R_b\le 2$; this gives $\Delta
  F_\mathbf{c}^{\mathrm{max}}(\mathbf{e};\beta)\le 2\wgt \mathbf{c}$.
  Minimizing over the vectors degenerate with $\mathbf{c}$
  gives the stated result.
\end{proof}

Note that at zero temperature and in the absence of disorder,
$\mathbf{e}=\mathbf{0}$, the upper bound in
Lemma~\ref{lemma:Fmax-bounds} is saturated.  We conjecture that a
similar asymptotic scaling, with some finite 
\begin{equation}
  \lambda_\mathbf{c}\equiv 
  {[\Delta
      F_\mathbf{c}^\mathrm{max}(\mathbf{s}_\mathbf{e};\beta)]\over
    d_\mathbf{c}},\label{eq:defect-tension}   
\end{equation}
should be valid for the free energy increments averaged over
disorder, with the defect \emph{tension} $\lambda_\mathbf{c}$
analogous to the domain wall tension in the 2D Ising model.  In the
fixed-defect phase, where $\Delta F_\mathbf{c}$ is expected to
diverge, we thus expect the tensions~\eqref{eq:defect-tension} to be
non-zero, $\lambda_\mathbf{c}>0$.

{\em In the defect-free phase:\/} In such a phase, the total partition
function~\eqref{eq:Z-tot} is entirely dominated by that without any
extended defects, see Eq.\ \eqref{eq:partition0}.  Instead of
Eq.\ \eqref{eq:defect-free-energy-max}, it is convenient to consider
the free energy increment for flipping the bonds corresponding to the codeword
$\mathbf{c}$ starting with a given defect configuration $\mathbf{e}$,
\begin{equation}
  \label{eq:defect-free-energy-zero}
  \Delta F_\mathbf{c}^{(0)}(\mathbf{e};\beta)\equiv  
  \beta^{-1} \log {Z_0(\mathbf{e};\beta)\over
    Z_\mathbf{c}(\mathbf{e};\beta)}. 
\end{equation}
Similar to the upper bound in Lemma \ref{lemma:Fmax-bounds}, we can
state 
\begin{equation}
  \label{eq:F0-bound-d}
    \Delta F_\mathbf{c}^{(0)}(\mathbf{e};\beta)\le 2d_\mathbf{c}; 
\end{equation}
however, the corresponding lower bound might be violated for some
disorder configurations $\mathbf{e}$ where
$\mathbf{c}_\mathrm{max}(\mathbf{e})\not\simeq\mathbf{0}$.  In the
defect-free phase, the total probability of such configurations,
$P_\mathrm{defect}$, as well as the configurations where
$F_\mathbf{c}^{(0)}(\mathbf{e};\beta)$ remains bounded,
$P_\mathrm{finite}$, should be vanishingly small at large $n$,
$P_\mathrm{defect}+P_\mathrm{finite}\to0$.  The corresponding bounds
can be readily formulated by analogy with
Lemma~\ref{lemma:upside-down}.  As a result, while in general the
increments in Eqs.\ \eqref{eq:defect-free-energy-max} and
\eqref{eq:defect-free-energy-zero} have both the initial and the final
states different and cannot be easily compared, in the defect-free
phase the corresponding averages should coincide asymptotically at
$n\to\infty$. In particular, this implies $[\Delta
  F_\mathbf{c}^{(0)}(\mathbf{e};\beta)]\to\infty$ at large $n$ in the
defect-free phase.

{\em On the Nishimori line:\/} 
According to Theorem  \ref{th:Nishimori-line-defect-free}, the only
ordered phase at the Nishimori line is the defect-free phase.  This
immediately gives Corollary~\ref{th:corollary}. 

On the Nishimori line, it is convenient to consider the free energy $\Delta
F_\mathbf{c}(\mathbf{s};\beta)$ of a defect $\mathbf{c}$ averaged over the
errors $\mathbf{e}$ with the same syndrome, $\mathbf{s}=\widetilde
G\mathbf{e}^T$,
\begin{equation}
  \label{eq:defect-free-energy-ave-s}
  \Delta F_\mathbf{c}(\mathbf{s};\beta)\equiv  \left[
  \Delta F_\mathbf{c}^{(0)}(\mathbf{e};\beta) \right]_\mathbf{s}, 
\end{equation}
where the average is extended over all non-equivalent codewords $\mathbf{c}$, 
\begin{equation}
  \label{eq:syndrome-averaging}
  \left[f(\mathbf{e})\right]_s\equiv \sum_\mathbf{c}
  {P_0(\mathbf{e}+\mathbf{c})\over 
    P_\mathrm{tot}(\mathbf{s})} f(\mathbf{e}+\mathbf{c}). 
\end{equation}
 For the average \eqref{eq:defect-free-energy-ave-s}, we prove the
 following version of  
Lemma~\ref{lemma:Fmax-bounds}:
\begin{lemma}
  \label{lemma:F-syndrome-averaged-Nishimori}
  At the Nishimori line, for every allowed syndrome $\mathbf{s}$ and
  every codeword $\mathbf{c}$, the free
  energy averaged over the errors with the same syndrome satisfies
  $0\le \Delta 
  F_\mathbf{c}(\mathbf{s};\beta_p)
  \le 2d_\mathbf{c} $. 
\end{lemma}
\begin{proof}%
  The upper bound is trivial since it applies for every term in the
  average, see Eq.\ \eqref{eq:F0-bound-d}.  The lower bound follows
  from the Gibbs inequality.  Explicitly, introduce two normalized
  distribution functions of codewords $\mathbf{b}$:
  $f_{\mathbf{b}}\equiv
  P_\mathbf{0}(\mathbf{e}')/P_\mathrm{tot}(\mathbf{s})$, $
  g_{\mathbf{b}}\equiv
  P_\mathbf{c}(\mathbf{e}')/P_\mathrm{tot}(\mathbf{s})$, where
  $\mathbf{e}'\equiv \mathbf{e}+\mathbf{b}$; then, using the
  map~\eqref{eq:nishimori-line} on the Nishimori line,
  $$
  \beta\Delta F_\mathbf{c}(\mathbf{s};\beta_p)=\sum_{\mathbf{b}}
  f_\mathbf{b}\log {f_\mathbf{b}\over g_\mathbf{b}}\ge
  \sum_\mathbf{b}f_\mathbf{b}\left(1-{g_\mathbf{b}\over f_\mathbf{b}}\right)=0,
  $$
  where the summation is done over all non-equivalent codewords
  $\mathbf{b}$ and we used $\log (x)\ge 1-1/x$.
\end{proof}
Note that this Lemma gives an alternative proof of Theorem
\ref{th:Nishimori-line-defect-free}.
\subsection{Self-averaging} 
\ Conditions of Theorem \ref{th:divergent-defect-energy-max} guarantee
that the disordered system is not in a spin glass phase.  A
self-averaging for the partition functions
$Z_\mathbf{c}(\mathbf{e};\beta)$ would immediately imply the statement
of the theorem.  Note however, that (\textbf{i}) in the presence of
disorder self-averaging is not expected for the partition function
even in the case of the toric codes as fluctuations could be
exponentially large, and (\textbf{ii}) spin models corresponding to
general families of quantum codes, whether LDPC or not, are expected
to involve highly non-local interactions. Thus, without additional
conditions, one cannot guarantee self-averaging even for the free
energy.  

However, we did not rely on self-averaging in any of the proofs.  In
particular, results in this section apply to spin models corresponding to
finite-rate quantum hypergraph-product and related
codes\cite{Tillich-Zemor-2009,Kovalev-Pryadko-2012} that can be
obtained from random binary LDPC codes:

\begin{example}  
  \label{ex:finite-R}
  This is a special case of the model in Example \ref{ex:TZ-vanilla}.
  Consider a \emph{random} binary matrix ${\cal H}$ with $h$ non-zero
  entries per row and $v$ per column, with $h<v$, e.g., see
  Ref.~\cite{Gallager-1962}.  The rate of the corresponding binary
  code ${\cal C}_{\cal H}^\perp$ with parameters $[n_c,k_c,d_c]$ is
  limited, $R_\mathrm{c}\equiv k_\mathrm{c}/n_\mathrm{c}\ge 1-h/v$.
  With high probability at large $n_\mathrm{c}$, the classical code
  will have the relative distance in excess of
  $\delta_\mathrm{c}\equiv\delta_c(h,v)$ given in
  Ref.~\cite{Gallager-1962}.  Such an $[n_{c},k_{c},d_{c}]$ code
  produces a quantum HP code \eqref{eq:Till} with ${\cal H}_1={\cal
    H}_2^T={\cal H}$, which is a quantum LDPC code with the asymptotic
  rate $k/n\ge (v-h)^{2}/(h^{2}+v^{2})$ and the distance scaling as
  $d/\sqrt{n}=\delta_{c}v/\sqrt{h^{2}+v^{2}}$.  Such a code has a
  decoding transition at a finite $p$, see
  Ref.~\cite{Kovalev-Pryadko-FT-2013}.  Our present results indicate
  that each of the corresponding spin models \eqref{eq:Z0} has
  non-local bonds involving up to $v$ spins, exponentially large
  number of mutually inequivalent extended defects, and an ordered
  state where such defects do not appear.  In addition, as already
  stated in Example \ref{ex:TZ-vanilla}, the two models
  are self-dual modulo logical operators.
\end{example}

\section{Phase transitions}
\subsection{Transition to a disordered phase}
\ {\em Transition mechanism:\/} 
An ordered phase (whether fixed-defect or defect-free) of the
model~\eqref{eq:Z-tot} is characterized by a unique defect pattern
$\mathbf{c}_\mathrm{max}(\mathbf{e})$ for every likely configuration
of flipped bonds $\mathbf{e}$.  In the case of a code family where $k$
remains fixed, for the stability of such a phase it is sufficient that
non-trivial defects $\mathbf{c}\not\simeq\mathbf{0}$ have divergent
free energies, as in Theorem \ref{th:divergent-defect-energy-max}.  On
the other hand, defects can proliferate if at least one of the free
energies $\Delta F_\mathbf{c}^{\rm max}$ remains bounded in the
asymptotic $n\to\infty$ limit.

The situation is different in the case of a code family with divergent
$k$, e.g., with fixed rate $R\equiv k/n$, as in
Example~\ref{ex:finite-R}.  Here, the number of %
different defects, $2^{2k}-1$, diverges exponentially at large $n$; in
an ordered phase the free energies of individual defects must be large enough
to suppress this divergence.  This implies, in particular, that for a
typical defect the
tension~\eqref{eq:defect-tension} must exceed certain limit.  The
statement of Theorem \ref{th:tension-average} concerns the
corresponding average tension,
\begin{equation}
  {\overline\lambda}\equiv(2^{2k}-1)^{-1}
  \sum_{\mathbf{c}\not\simeq\mathbf{0}}\lambda_\mathbf{c}. 
  \label{eq:average-tension}
\end{equation}

\begin{proof}{of Theorem \ref{th:tension-average}.}
Let us start with a version of Lemma~\ref{lemma:upside-down} for the
fixed-defect phase (Def.\ \ref{def:fixed}): for
  any likely disorder configuration $\mathbf{e}$,
  \begin{equation}
    \sum_{\mathbf{c}\not\simeq\mathbf{0}}{Z_{\mathbf{c}+\mathbf{c}_\mathrm{max}(\mathbf{e})}
      (\mathbf{e};\beta) \over 
      Z_\mathrm{max}(\mathbf{s_\mathbf{e}};\beta)}\to0,
    \label{eq:large-sum}
  \end{equation}
  asymptotically at $n\to\infty$.  Note that we cannot just average
  this expression term-by-term, since unlikely errors
  could potentially dominate the sum which involves an exponentially
  large number of terms.  Instead, we fix some $\epsilon>0$ and first
  consider the average of Eq.\ \eqref{eq:large-sum} only over the
  ``good'' errors where the sum does not exceed $\epsilon$.  Using the
  standard inequality $\exp\langle f\rangle\le \langle \exp f\rangle
  $, we obtain the following expression involving the averages of the
  free energies \eqref{eq:defect-free-energy-max} over ``good'' errors
  only:
  \begin{equation}
    \label{eq:large-sum-upper-bound}
    \sum_{\mathbf{c}\not\simeq\mathbf{0}}\exp\left({-\beta [\Delta F_\mathbf{c}^{\rm
            max}(\mathbf{s}_\mathbf{e};\beta)]_\mathrm{good}}\right)\le \epsilon.
  \end{equation}
  Rewriting this sum in terms
  of an average over non-trivial defects which we denote as $\left\langle\,
  \cdot\,\right\rangle_{\mathbf{c}\not\simeq\mathbf{0}}$, and using the same 
  inequality, we get 
  \begin{equation}
    \label{eq:large-sum-upper-bound}
    (2^{2k}-1)\exp\left(-\beta \left\langle [\Delta
            F_\mathbf{c}^\mathrm{max}(\mathbf{s}_\mathbf{e};\beta)]_\mathrm{good}
        \right\rangle_{\mathbf{c}\not\simeq\mathbf{0}}\right)\le 
    \epsilon. 
  \end{equation}
  It is convenient to introduce an analog of the 
  tension~\eqref{eq:defect-tension} for finite $\epsilon$,
  \begin{equation}
    \label{eq:tension-epsilon}
    \lambda_\mathbf{c}^{(\epsilon)}\equiv {[\Delta F_\mathbf{c}^{\rm
          max}]_\mathrm{good}\over d_\mathbf{c}}, 
  \end{equation}
  along with the corresponding average
  ${\overline\lambda}_{(\epsilon)}$ over non-trivial defects
  $\mathbf{c}\not\simeq\mathbf{0}$, defined as in Eq.\
  \eqref{eq:average-tension}.  According to
  Lemma~\ref{lemma:Fmax-bounds}, each of the tensions satisfy $0\le
  \lambda_\mathbf{c}^{(\epsilon)}\le 2$, which means the same bounds
  for the defects-average, $0\le {\overline\lambda}_{(\epsilon)}\le 2$.  With
  the help of the trivial upper bound $d_\mathbf{c}\le
  N_\mathrm{b}=2n$, Eq.\ \eqref{eq:large-sum-upper-bound} gives
  \begin{equation}
    \label{eq:large-sum-bound}
    (2^{2k}-1)\exp({-2n\beta{\overline\lambda}^{(\epsilon)}})\le\epsilon, 
  \end{equation}
  which implies for large $n$, $k$
  \begin{equation}
    \label{eq:tension-bound-epsilon}
    \beta{\overline\lambda}^{(\epsilon)}\ge {k\over n}\log 2=R\log2.
  \end{equation}
  We can now introduce the full average tension ${\overline\lambda}$
  which involves both ``good'' and ``bad'' errors by writing a
  Bayesian expansion similar to Eq.\ \eqref{eq:Bayesian-expansion}.
  The key observation leading to the statement of the Theorem is that
  the contribution of ``bad'' errors disappears in the large-$n$ limit
  since for each error configuration the tension is limited, while the
  total probability of ``bad'' errors $P_\mathrm{bad}\to0$.
\end{proof}
As a consequence, for any code family with a finite rate $R$, we
expect one of the two possibilities at the transition to a disordered
phase: (\textbf{i}) Transition driven by proliferation of some (e.g.,
finite) subset of the defects whose tensions $\lambda_\mathbf{c}$
vanish at the transition, with the average in Theorem
\ref{th:tension-average} still finite; and (\textbf{ii}) Transition
driven by the entropy of some macroscopic number of the defects, in
which case tensions of all defects remain bounded at the
transition, $\lambda_\mathbf{c}\ge\lambda_0>0$.  In the case
(\textbf{i}), one gets to a phase with ``limited disorder'' where only
some of all possible defects $\mathbf{c}$ may happen with non-zero
probability at large $n$.

{\em Continuity of the transition:\/} At the Nishimori line, the
average energy is known
exactly\cite{Nishimori-1981,Nishimori-1980,Nishimori-book}, it is a
continuous function of parameters.  This guarantees the continuity of
the decoding transition.  The same conclusion can be drawn from the
bound \eqref{eq:specific-heat-bound} on the specific heat along the
Nishimori line---the derivation is identical to the standard
case\cite{Nishimori-1980,Morita-Horiguchi-1980,Nishimori-1981,Nishimori-book}.

On the other hand, away from the Nishimori line, the transition from
an ordered to a disordered phase can be (and often is) discontinuous.
In particular, mean field analysis using the TAP equations (named for
Thouless, Anderson, and Palmer, see
Ref.~\cite{Thouless-Anderson-Palmer-1977}) generically gives a
discontinuous transition for local magnetization whenever the bonds
$R_b$ couple more than two spins.

{\em Self-duality in the absence of disorder:\/}
In the absence of errors, we can use Wegner's
duality~\eqref{eq:duality} to relate the partition
functions of the models with the generator matrices $G$ and $G^*$,
that is, Eqs.\ \eqref{eq:partition0} and \eqref{eq:Z-tot}, since the
matrices $G^*$ and $\widetilde G^*$ differ by an inessential
permutation of columns (bonds).  Assuming the transition is unique,
whether continuous or not, it must happen at the
self-dual point, $\sinh
(2\beta_{\rm s.d.})=1$.  Here Eq.\ \eqref{eq:duality} gives $Z_{\rm
  tot}(\mathbf{0};\beta_{\rm s.d.})= 2^k Z_{0}(\mathbf{0};\beta_{\rm
  s.d.})$, or, equivalently,
\begin{equation}
  \label{eq:clean-self-dual-point}
  \sum_{\mathbf{c}\not\simeq\mathbf{0}}e^{-\beta_{\rm s.d.}\Delta
    F^{(0)}_\mathbf{c}(\mathbf{0};\beta_{\rm s.d.})}=2^k-1.
\end{equation}
This equation is exact since no disorder is involved.  The summation
over ${\bf c}$ here includes $2^{2k}-1$ terms, and the result is
independent of the distance of the code.  For a finite-$R$ code
family, arguments similar to those in the proof of Theorem
\ref{th:tension-average} give a lower bound
${\overline\lambda}_\mathrm{s.d.}\ge (R/2)\ln 2$, which is smaller by
half of the corresponding bound deep inside an ordered phase.

{\em Location of the multicritical point:\/} In many types of local spin
glasses on self-dual lattices the transition from the ordered phase on
the Nishimori line happens at a multicritical point whose location to
a very good accuracy has been predicted by the strong-disorder
self-duality
conjecture\cite{Nishimori-1979,Nishimori-Nemoto-2003,Nishimori-2007,%
  Nishimori-Ohzeki-2006,Ohzeki-Nishimori-Berker-2008,%
  Ohzeki-2009,Bombin-PRX-2012,Ohzeki-Fujii-2012}.  In case of the
Ising spin glasses, the corresponding critical probability $p_c\approx
0.110$ satisfies Eq.\ \ref{eq:self-duality-disordered}.  The
derivation of this expression\cite{Nishimori-1979} uses explicitly
only the probability distribution of allowed energy values for a
single bond.  Our limited simulations indicate that for several
quasi-local models (see Example~\eqref{ex:extended-cyclic}) with finite
$k$ the multicritical point is indeed located at $p_c\approx 0.11$,
also very close to the Gilbert-Varshamov existence bound for zero-rate
codes.  However, for code families with finite rates $k/n$, see
Example \ref{ex:finite-R}, the threshold probability must be below the
Shannon limit \eqref{eq:shannon-threshold-CSS}, which means the
self-duality conjecture must be strongly violated for $R>1/2$.  

\subsection{Transition between defect-free and fixed-defect phases}
\ Theorem~\ref{th:Nishimori-line-defect-free} states that on the
Nishimori line below the decoding transition the spin
model \eqref{eq:partition0} is in the defect-free phase.  If a
distinct fixed-defect phase exists somewhere on
the phase diagram, there is 
a possibility for a transition between these phases. 

More generally, defect-free phase is a special case of an ordered
fixed-defect phase.  One can imagine a transitions between two such
phases.  However, at least in the case of a temperature-driven
transition, the spin model \eqref{eq:partition0} must become
disordered at the transition point.  Indeed, for a transition to
happen at $T=T_0(p)$, at least for some of the likely disorder
configurations, for $T<T_0(p)$, $Z_{\mathbf{c}_1}(\mathbf{e};\beta)$
must dominate, while for $T>T_0(p)$, some of errors $\mathbf{e}$ will
be dominated by $Z_{\mathbf{c}_2}(\mathbf{e};\beta)$ with
$\mathbf{c}_2\not\simeq\mathbf{c}_1$.  This implies that at the actual
transition point some codewords must become degenerate with non-zero
probability, which would violate the condition in Def.\
\ref{def:fixed}.  Once the system becomes disordered at some $p$, one
would generically expect it to remain disordered at larger $p$.  By
this reason, we expect that non-trivial fixed-defect phases are not
common.


\subsection{Absence of a local order parameter}
\ In Examples \ref{ex:CSS} to \ref{ex:finite-R} we considered some
spin models which do not have any gauge-like symmetries.  However, the
same approach can be also used to construct non-local spin models
which have ``local'' gauge symmetries and at the same time highly 
non-trivial phase diagrams. 

The following example is a generalization of the mutually dual
three-dimensional Ising model and a random plaquette  $\Z2$ gauge.
\begin{example}
\label{ex:gauge}
Consider a CSS
code \eqref{eq:CSS} with the generators: 
\begin{eqnarray}
  {\cal G}_{X}&=& \left( E_{1}\otimes G,\;\;\; R\otimes E_{2}\right),
  \label{eq:3D1}\\
  {\cal G}_{Z}&=& \left(\begin{array}{cc}
      R\otimes\widetilde{E}_{2},&\widetilde{E}_{1}\otimes G \\
      E_{1}\otimes \widetilde{G},&0
    \end{array}\right),
  \label{eq:3D2}
\end{eqnarray}
where $R$ is a square circulant matrix corresponding to the polynomial
$h(x)=1+x$ and $G\equiv (G_X,G_Z)$ is the generator matrix
\eqref{eq:generator-matrix} of an arbitrary quantum code. This
construction follows the hypergraph-product code construction
\eqref{eq:Till}, and the unit matrices $E_{1}$, $\widetilde{E}_{1}$,
$E_{2}$, $\widetilde{E}_{2}$ are chosen accordingly. The additional
block involving the conjugate matrix $\widetilde{G}=(G_{Z},G_{X})$
differentiates this construction from the hypergraph-product code
construction.  This code defines two non-interacting, mutually dual
spin models \eqref{eq:Z0}.  In particular, when $G$ corresponds to a
toric code, we recover a three dimensional Ising model for $\mu=X$,
and a three dimensional random plaquette $\Z2$ gauge model for
$\mu=Z$.
\end{example}
A spin model with a local gauge symmetry cannot have a local order
parameter\cite{Wegner71}.  Thus, one cannot hope to construct a local
order parameter that would describe the transition from a defect-free
phase and be applicable to all of the models \eqref{eq:partition0}.  

The same result can be obtained by noticing that the transition from
the defect-free phase can be driven by delocalization of any of
$2^{2k}-1$ non-trivial defects.  For a finite-$R$ code family this
number scales exponentially with $n$; we find it not likely that an
order parameter defined locally can distinguish this many
possibilities.  

\subsection{Spin correlation functions} 
\ The average of any product of spin variables which cannot be
expressed as a product of the bond variables in the Hamiltonian is
zero \cite{Wegner71}.  Thus, we consider two most general non-trivial
spin correlation functions:
\begin{eqnarray}
  \label{eq:correlation-function-tot}
  Q^\mathbf{m}_\mathrm{tot}(\mathbf{e};\beta)&\equiv&
  {\mathscr{Z}_{\mathbf{e},\mathbf{m}}(\widetilde
    G^*;\{K_b=\beta\})\over \mathscr{Z}_{\mathbf{e},\mathbf{0}}(\widetilde
    G^*;\{K_b=\beta\})} ,\\
  Q^\mathbf{m}_\mathbf{c}(\mathbf{e};\beta)&\equiv&
  {\mathscr{Z}_{\mathbf{e}+\mathbf{c},\mathbf{m}}(
    G;\{K_b=\beta\})\over
    \mathscr{Z}_{\mathbf{e}+\mathbf{c},\mathbf{0}}(G;\{K_b=\beta\})};  
  \label{eq:correlation-function-c}
\end{eqnarray}
both correlation functions satisfy $-1\le Q^\mathbf{m}(\mathbf{e};\beta)\le 1$. 
The thermal average in Eq.\ \eqref{eq:correlation-function-c}
corresponds to summation over spin configurations in
$Z_\mathbf{c}(\mathbf{e};\beta)$, while that in Eq.\
\eqref{eq:correlation-function-tot} corresponds to the same defect and
spin configurations that enter $Z_\mathrm{tot}(\mathbf{s};\beta)$,
cf.\ Eq.\ \eqref{eq:Z-tot}.  Using the explicit form
\eqref{eq:generalized-wegner}, definitions of $Z_\mathrm{tot}$ and
$Z_\mathbf{c}$, and the fact that additional linearly-independent rows
in $\widetilde G^*$ form a basis of non-equivalent codewords
$\mathbf{c}$, we can write the following expansion
\begin{equation}
  \label{eq:correlation-function-expansion}
  Q_\mathrm{tot}^\mathbf{m}(\mathbf{e};\beta)
  =\sum_\mathbf{c}(-1)^{\mathbf{c}\cdot \mathbf{m}}
  {Z_\mathbf{c}(\mathbf{e};\beta)\,
    Q_\mathbf{c}^\mathbf{m}(\mathbf{e};\beta)\over
      Z_\mathrm{tot}(\mathbf{s}_\mathbf{e};\beta)}.
\end{equation}
The correlation functions contain the products of $\prod_b R_b
^{m_b}=\prod_r (S_r )^{G_{rb}m_b}$, or the product of spin variables
in the support of the syndrome vector $\mathbf{s}_{\widetilde{\mathbf{m}}}\equiv G
\mathbf{m}^T=\widetilde G\widetilde{\mathbf{m}}^T$ corresponding to
$\mathbf{m}$.  Thus, the defined 
correlation functions are trivially symmetric with respect to any
gauge symmetries, $S_r\to S_r (-1)^{\alpha_r}$, $\boldsymbol\alpha
G=0$ (present whenever there are $N_g>0$ linearly dependent rows of
$G$), as well as the transformations of $\mathbf{m}$ leaving the
syndrome invariant, $\mathbf{m}\to \mathbf{m}+{\boldsymbol\gamma}
\widetilde G$.  

\emph{Wilson loop:} In lattice gauge theory, in the absence of a local
order parameter, the deconfining transition can be characterized by
the average of the Wilson loop operator\cite{Wilson-loop-1974}, with
the thermal and disorder average scaling down as an exponent of the
area in the high-temperature phase, and an exponent of the perimeter
in the low-temperature phase.  In the case of the three-dimensional
${\mathbb Z}_2$ gauge model\cite{Wegner71,Kogut79}, see Example
\ref{ex:gauge}, the corresponding correlator is a product of plaquette
operators covering certain surface.  The correlation function
\eqref{eq:correlation-function-c} is a natural generalization to
non-local Ising models, with the minimum weight
$\mathbf{d}_\mathbf{m}\equiv
\min_{\boldsymbol\gamma}\wgt(\mathbf{m}+{\boldsymbol\gamma}\tilde G)$
of $\mathbf{m}$ corresponding to the area, and the binary weight of
the syndrome $\mathbf{s}_{\widetilde{\mathbf{m}}}$ corresponding to
the perimeter.  Indeed, taking $\mathbf{e}=\mathbf{c}=\mathbf{0}$, at
high temperatures, independent bond variables $R_b$ fluctuate
independently, and one can write
$Q_\mathbf{0}^\mathbf{m}(\mathbf{0};\beta)=\langle \prod
R_b^{m_b}\rangle\propto \beta^{d_\mathbf{m}}$, which corresponds to
the area law.  The same quantity at low temperatures can be evaluated
in leading order by substituting average spin $S_b\to \langle
S_b\rangle\sim M$, with the result
$Q_\mathbf{0}^\mathbf{m}(\mathbf{0};\beta)\propto M^{\wgt
  \mathbf{s}_{\widetilde{\mathbf{m}}}}$, the perimeter law.  We expect
  such a behavior to persist in a finite range of temperatures below
  the transition from the ordered phase, at least in the case of LDPC
  codes.

However, in general there is no guarantee that the spin model
\eqref{eq:partition0} has a unique transition, and the functional
form of the spin correlation function
\eqref{eq:correlation-function-c} with generic $\mathbf{m}$ cannot
be easily found at intermediate temperatures.  By this reason, it
remains an open question whether the scaling of the analog of the
Wilson loop can be used to distinguish between specific disordered
phases.

\emph{Indicator correlation functions.}  Consider the correlation
function \eqref{eq:correlation-function-expansion} for $\mathbf{m}$
such that the corresponding syndrome is zero
$\mathbf{s}_{\widetilde{\mathbf{m}}}=\mathbf{0}$.  Then the spin
products in each term of the expansion disappear, and
$Q^\mathbf{m}_\mathbf{c}(\mathbf{e};\beta)=1$ for any $\mathbf{c}$.
The corresponding $\mathbf{m}$ are just the dual codewords
$\widetilde{\mathbf{b}}$.  In general, for a pair of codewords
$\mathbf{b}$, $\mathbf{c}$, the scalar product $\mathbf{c}\cdot
\widetilde{\mathbf{b}}=0$ iff the corresponding logical operators
commute, see the \textbf{Background} section.  For each codeword
$\mathbf{c}\not\simeq\mathbf{0}$ there is at least one codeword
$\mathbf{c}'$ such that $\mathbf{c}\cdot{\widetilde{\mathbf{c}}}'=1$,
and the $2k$ scalar products $\mathbf{c}\cdot
{\widetilde{\mathbf{b}}}$ with the basis codewords $\mathbf{b}$ are
sufficient to recover the equivalence class of $\mathbf{c}$.

We further note that in the defect-free phase, for any likely disorder
$\mathbf{e}$, $Z_\mathrm{tot}(\mathbf{s}_\mathbf{e};\beta)$ is
dominated by the term with $\mathbf{c}=0$, thus at large $n$ the
average
$[Q_\mathrm{tot}^{\widetilde{\mathbf{b}}}(\mathbf{s}_\mathbf{e};\beta)]=1$
for any codeword $\mathbf{b}$.  Similarly, in a fixed-defect phase,
there is only one dominant term $Z_\mathbf{c}(\mathbf{e};\beta)$, and
$[Q_\mathrm{tot}^{\widetilde{\mathbf{b}}}(\mathbf{s}_\mathbf{e};\beta)]=\pm1$;
the patterns of signs for different $\mathbf{b}$ can be used to find
out which of the codewords $\mathbf{c}$ dominates the partition
function.

\subsection{Bound on the location of the defect-free phase}
In order to prove the Theorem \ref{th:boundary}, we first need to
extend identities of Nishimori's gauge theory of spin
glasses\cite{Nishimori-1981,Horiguchi-Morita-1981,Nishimori-book} to
the averages of the 
spin correlation functions \eqref{eq:correlation-function-tot}.  We
prove the following
\begin{lemma}
  \label{lemma:Nishimori-identities}
  The disorder average of the spin correlation function
  \eqref{eq:correlation-function-tot} for any $\mathbf{m}$ satisfies 
  $[Q_\mathrm{tot}^\mathbf{m}(\mathbf{e};\beta)]=[Q_\mathrm{tot}^\mathbf{m}(\mathbf{e};\beta)\,Q_\mathrm{tot}^\mathbf{m}(\mathbf{e};\beta_p)]$.
\end{lemma}
\begin{proof}
    Follows exactly the proof in the usual case, if we observe 
    $$
    \sum_\mathbf{{\boldsymbol\alpha}}
    P_0(\mathbf{e}+{\boldsymbol\alpha}\widetilde G^*)
    =2^{N_r-N_g+N_g^*}Z_\mathrm{tot}(\mathbf{s}_\mathbf{e};\beta_p),
    $$
    where $N_r$ is the number of rows of the matrix $G$. 
\end{proof}
\begin{proof}{of Theorem \ref{th:boundary}.}
  To shorten the notations, denote the correlation function in Lemma
  \ref{lemma:Nishimori-identities} as $A\equiv
  Q_\mathrm{tot}^\mathbf{m}(\mathbf{e};\beta)$ and $B$ the same
  correlation function at the Nishimori temperature, $\beta=\beta_p$.
  Lemma \ref{lemma:Nishimori-identities} gives
\begin{equation}
  \label{eq:sg-identities}
  [A]=[AB],\quad [B]=[B^2].
\end{equation}
Now, for any real-valued $t$, the inequality
\begin{equation}
  0\le [(A-t B)^2]=[A^2]+t^2[B^2]-2t [A B]
  \label{eq:uncertainty}
\end{equation}
must be valid.  This is equivalent to $[AB]^2\le [A^2][B^2]$.  Using
the identities~\eqref{eq:sg-identities}, we obtain $[A]^2\le
[A^2][B]\le [B]=[B^2]$, which is equivalent to
\begin{equation}
  \label{eq:sg-inequality}
  [Q^{\mathbf{m}}_\mathrm{tot}(\mathbf{e};\beta)]^2\le [Q^{\mathbf{m}}_\mathrm{tot}(\mathbf{e};\beta_p)].
\end{equation}
A different derivation of this inequality can be found in
Ref.~\cite{Nishimori-2002}.  
If sum both sides of Eq.\ \eqref{eq:sg-inequality} over all dual codewords
$\mathbf{m}=\widetilde{\mathbf{c}}$, using the expansion
\eqref{eq:correlation-function-expansion}, we obtain
\begin{equation}
  \label{eq:sg-inequality-two}
 \sum_{\mathbf{c}}
 [Q^{\mathbf{m}=\widetilde{\mathbf{c}}}_\mathrm{tot}(\mathbf{e};\beta)]^2\le
 2^{2k}\left[{Z_0(\mathbf{e};\beta_p)\over
       Z_\mathrm{tot}(\mathbf{s}_\mathbf{e};\beta_p)}\right]. 
\end{equation}
The r.h.s.\ equals exactly the average probability of successful
decoding times $2^{2k}$; for large $n$ it equals $2^{2k}$ below the
decoding transition, $p<p_c$, and it is smaller than $2^{2k}$ above
the decoding transition.  On the other hand, we saw that in in the
defect-free phase, at large $n$, all correlation functions
$[Q^{\widetilde{\mathbf{m}}}_\mathrm{tot}(\mathbf{e};\beta)]=1$.
According to Eq.\ \eqref{eq:sg-inequality-two}, this is only possible
for $p<p_c$.  
\end{proof}
This implies that the phase boundary below the Nishimori line is
either vertical or reentrant as a function of temperature.  Recent
numerical studies suggest that the second option is true for the
random bond Ising model\cite{Thomas-Katzgraber-2011}.

\section{Concluding Remarks}

In this work we considered spin glass models related to the decoding
transition in stabilizer error correcting codes.  Generally, these are
non-local models with multi-spin couplings, with exact Wegner-type
self-duality at zero disorder, but no $S\to-S$ symmetry or other
sources of ground state degeneracy.  Nevertheless, we show that for
models corresponding to code families with maximum-likelihood decoding
(ML) transition at a finite bit error probability $p_c$, there is a
region of an ordered phase which must be limited to $p\le p_c$, and a
line of non-trivial phase transitions.

The models support generally non-topological extended defects which
generalize the notion of domain walls in local spin models.  For a
quantum code that encodes $k$ qubits, there are $2^{2k}-1$ different
types of extended defects.  A disordered phase is associated with
proliferation of at least one of such defects.  In an ordered phase,
the free energy of each defect must diverge at large $n$.  Moreover,
for a code family with finite rate $k/n$, the average defect tension,
an analog of domain wall line tension, must exceed some finite
threshold (Theorem \ref{th:tension-average}).

The original decoding problem corresponds to the Nishimori line at the
phase diagram of the disordered spin model, with the
maximum-likelihood (ML) decoding transition located exactly at the
multicritical point of the spin model.  The ML decoding threshold is
the maximum possible threshold for any decoder.  Thus, exploring this
connection with statistical mechanics of spin glasses, one can compare
codes irrespectively of the decoder efficiency, and get an absolute
measure of performance for any given, presumably suboptimal, decoder.


There are a number of open question in relation to the models we
studied.  In particular, is there some sort of universality for
transitions with nonlocal spin couplings?  If yes, what determines the
universality class, and is there an analog of the hyperscaling
relation?


\section*{Acknowledgments}
This work was supported in part by the U.S. Army Research Office under
Grant W911NF-11-1-0027, by the NSF under Grant 1018935,
and Central Facilities of the Nebraska Center for Materials and
Nanoscience supported by the Nebraska Research Initiative.
We acknowledge hospitality of the Institute for Quantum Information
and Matter, an NSF Physics Frontiers Center with support of the Gordon
and Betty Moore Foundation.  We also thank Professor Hidetoshi
Nishimori for useful comments on early version of the manuscript.

\bibliographystyle{pnas}
\bibliography{qc_all,sg,more_qc,corr,lpp,MyBIB}




\end{document}